\documentclass{article}
\usepackage{kr}

\usepackage{times}
\usepackage{soul}
\usepackage{url}
\usepackage[hidelinks]{hyperref}
\usepackage[utf8]{inputenc}
\usepackage[small]{caption}
\usepackage{graphicx}
\usepackage{amsmath}
\usepackage{amsthm}
\usepackage{booktabs}
\usepackage{algorithm}
\usepackage{algorithmic}
\usepackage{xspace}
\usepackage{amssymb}
\usepackage{mathtools}
\usepackage{todonotes}
\usepackage{multirow}

\newtheorem{example}{Example}
\newtheorem{theorem}{Theorem}
\newtheorem{definition}{Definition}
\newtheorem{proposition}{Proposition}

\newtheorem{lemma}{Lemma}

\usepackage{comment}
\includecomment{tr}
\excludecomment{paper}

\begin{paper}
    \title{Do Repeat Yourself:\\Understanding Sufficient Conditions for Restricted Chase Non-Termination}
\end{paper}
\begin{tr}
    \title{Do Repeat Yourself:\\Understanding Sufficient Conditions for Restricted Chase Non-Termination \\ (Technical Report) }
\end{tr}

\author{%
    Lukas Gerlach$^1$\and
    David Carral$^2$ \\
    \affiliations
    $^1$Knowledge-Based Systems Group, TU Dresden, Dresden, Germany\\
    $^2$LIRMM, Inria, University of Montpellier, CNRS, Montpellier, France \\
    \emails
    lukas.gerlach@tu-dresden.de,
    david.carral@inria.fr
}

\usepackage{amssymb}
\usepackage{amsthm}
\usepackage{enumerate}
\usepackage{mathtools}
\usepackage{todonotes}
\usepackage{xspace}

\makeatletter
\newtheorem*{rep@theorem}{\rep@title}
\newcommand{\newreptheorem}[2]{%
    \newenvironment{rep#1}[1]{%
        \def\rep@title{#2 \ref{##1}}%
        \begin{rep@theorem}}%
        {\end{rep@theorem}}}
        \makeatother

        \newreptheorem{theorem}{Theorem}
        \newreptheorem{lemma}{Lemma}
        \newreptheorem{corollary}{Corollary}

\newcommand{\FormatEntitySet}[1]{\ensuremath{\mathtt{#1}}}
\newcommand{\FormatFormulaSet}[1]{\ensuremath{\mathcal{#1}}}
\newcommand{\FormatFunction}[1]{\ensuremath{\mathsf{#1}}}
\newcommand{\FormatPredicate}[1]{\ensuremath{\mathsf{#1}}}
\newcommand{\FP}[1]{\FormatPredicate{#1}}

\newcommand{\FirstItem}{\textit{(i)}\xspace}
\newcommand{\SecondItem}{\textit{(ii)}\xspace}
\newcommand{\ThirdItem}{\textit{(iii)}\xspace}

\newcommand{\Tuple}[1]{\ensuremath{\langle{#1}\rangle}\xspace}

\renewcommand{\vec}[1]{\boldsymbol{#1}}

\newcommand{\FunctionSymbols}{\FormatEntitySet{Funs}\xspace}
\newcommand{\Funs}{\FunctionSymbols}

\newcommand{\Variables}{\FormatEntitySet{Vars}\xspace}
\newcommand{\Vars}{\Variables}
\newcommand{\Predicates}{\FormatEntitySet{Preds}\xspace}
\newcommand{\Preds}{\Predicates}

\newcommand{\Constants}{\FormatEntitySet{Cons}\xspace}
\newcommand{\C}{\Constants}
\newcommand{\Terms}{\FormatEntitySet{Terms}\xspace}

\newcommand{\EntitiesIn}[2]{\ensuremath{#1(#2)}\xspace}
\newcommand{\EI}[2]{\EntitiesIn{#1}{#2}}
\newcommand{\Arity}{\ensuremath{\FormatFunction{ar}}\xspace}
\newcommand{\Depth}{\ensuremath{\FormatFunction{depth}}\xspace}
\newcommand{\Max}{\ensuremath{\FormatFunction{max}}\xspace}
\newcommand{\Subterms}{\ensuremath{\FormatFunction{subterms}}\xspace}
\newcommand{\Vx}{{\ensuremath{\vec{x}}}\xspace}
\newcommand{\Vy}{\ensuremath{\vec{y}}\xspace}
\newcommand{\Vz}{\ensuremath{\vec{z}}\xspace}

\newcommand{\Vt}{{\ensuremath{\vec{t}}}\xspace}
\newcommand{\Vs}{\ensuremath{\vec{s}}\xspace}

\newcommand{\Vc}{\ensuremath{\vec{c}}\xspace}

\newcommand{\Formula}{\ensuremath{\upsilon}\xspace}
\newcommand{\Expression}{\ensuremath{\phi}\xspace}

\newcommand{\Rule}{\ensuremath{\rho}\xspace}
\newcommand{\RuleAux}{\ensuremath{\psi}\xspace}
\newcommand{\Query}{\ensuremath{\gamma}\xspace}
\newcommand{\Body}{\ensuremath{\beta}\xspace}
\newcommand{\BodyFunction}{\FormatFunction{body}\xspace}
\newcommand{\BF}{\BodyFunction}
\newcommand{\Head}{\ensuremath{\eta}\xspace}
\newcommand{\HeadFunction}{\FormatFunction{head}\xspace}
\newcommand{\HF}{\HeadFunction}
\newcommand{\Frontier}{\FormatFunction{frontier}\xspace}
\newcommand{\BranchFactor}{\FormatFunction{branching}\xspace}
\newcommand{\FactSet}{\ensuremath{\FormatFormulaSet{F}}\xspace}
\newcommand{\F}{\FactSet}
\newcommand{\FactSetAux}{\ensuremath{\FormatFormulaSet{G}}\xspace}
\newcommand{\FAux}{\FactSetAux}
\newcommand{\Instance}{\ensuremath{\FormatFormulaSet{I}}\xspace}

\newcommand{\Database}{\ensuremath{\FormatFormulaSet{D}}\xspace}
\newcommand{\D}{\Database}
\newcommand{\Rules}{\ensuremath{\FormatFormulaSet{R}}\xspace}
\newcommand{\R}{\Rules}
\newcommand{\KB}{\ensuremath{\FormatFormulaSet{K}}\xspace}
\newcommand{\ChaseTree}{\ensuremath{T}\xspace}
\newcommand{\CT}{\ChaseTree}

\newcommand{\Substitution}{\ensuremath{\sigma}\xspace}
\newcommand{\Subs}{\Substitution}
\newcommand{\SubstitutionAux}{\ensuremath{\tau}\xspace}
\newcommand{\SubsAux}{\SubstitutionAux}
\newcommand{\SkolemFunction}[2]{\ensuremath{f^{#1}_{#2}}}
\newcommand{\SF}[2]{\SkolemFunction{#1}{#2}}

\newcommand{\Skolemise}{\ensuremath{\FormatFunction{sk}}}

\newcommand{\Trigger}{\ensuremath{\lambda}\xspace}

\newcommand{\TriggerSeq}{\ensuremath{\Lambda}\xspace}

\newcommand{\Output}{\ensuremath{\FormatFunction{out}}\xspace}
\newcommand{\LabelFacts}{\FormatFunction{fct}\xspace}
\newcommand{\LF}{\LabelFacts}
\newcommand{\LabelTrigger}{\FormatFunction{trg}\xspace}
\newcommand{\LT}{\LabelTrigger}

\newcommand{\MFA}{\text{MFA}\xspace}

\newcommand{\RMFA}{\text{RMFA}\xspace}
\newcommand{\RMFAFunction}{\FormatFunction{RMFA}\xspace}
\newcommand{\RMFAF}{\RMFAFunction}

\newcommand{\BirthFacts}[2]{\ensuremath{\FormatFunction{BirthF}_{#1}(#2)}\xspace}
\newcommand{\Skeleton}[2]{\ensuremath{\FormatFunction{skeleton}_{#1}(#2)}\xspace}

\newcommand{\DMFC}{\text{DMFC}\xspace}

\newcommand{\RMFC}{\text{RMFC}\xspace}

\newcommand{\DRPC}{\text{DRPC}\xspace}
\newcommand{\DRPCFunction}{\FormatFunction{DRPC}\xspace}
\newcommand{\DRPCF}{\DRPCFunction}
\newcommand{\RPC}{\text{RPC}\xspace}
\newcommand{\RPCFunction}{\FormatFunction{RPC}\xspace}
\newcommand{\RPCF}{\RPCFunction}

\newcommand{\RuleDatabase}[1]{\ensuremath{\D_{#1}}\xspace}

\newcommand{\UCSubs}{\ensuremath{\Subs_{\textit{uc}}}\xspace}

\newcommand{\HeadChoice}{\FormatFunction{hc}\xspace}

\newcommand{\Branch}{\FormatFunction{branch}\xspace}

 \newcommand{\DoubleExpTime}{\textsc{2ExpTime}\xspace}

 \allowdisplaybreaks

\begin{document}

\maketitle

\begin{abstract}
  The disjunctive restricted chase is a sound and complete procedure for solving boolean conjunctive query entailment over knowledge bases of disjunctive existential rules.
  Alas, this procedure does not always terminate and checking if it does is undecidable.
  However, we can use acyclicity notions (sufficient conditions that imply termination) to effectively apply the chase in many real-world cases.
  To know if these conditions are as general as possible, we can use cyclicity notions (sufficient conditions that imply non-termination).
  In this paper, we discuss some issues with previously existing cyclicity notions, propose some novel notions for non-termination by dismantling the original idea, and empirically verify the generality of the new criteria.
\end{abstract}

\section{Introduction}

The \emph{(disjunctive) chase} \cite{DBLP:journals/tods/BourhisMMP16}
is a sound and complete
bottom-up materialization procedure to reason with \emph{knowledge bases} (KBs) featuring
\emph{(disjunctive existential) rules}.
In some cases, we can apply the chase to determine if a conjunctive query or a fact is a consequence of a KB under standard first-order semantics.

\begin{example}\label{exp:main}
Consider the KB $\KB = \Tuple{\R, \D}$ where $\R$ is the rule set $\{(\ref{main1}\text{--}\ref{main4})\}$ and $\D$ is the database $\{\FP{Engine}(d)\}$.
\begin{align}
\FP{Engine}(x) &\to \big(\exists  v . \FP{IsIn}(x, v) \land \FP{Bike}(v)\big) \lor \FP{Spare}(x) \label{main1}\\
\FP{Bike}(x) &\to \exists w . \FP{Has}(x, w) \land \FP{Engine}(w) \label{main2}\\
\FP{IsIn}(x, y) &\to \FP{Has}(y, x) \label{main3} \\
\FP{Has}(x, y) &\to \FP{IsIn}(y, x) \label{main4}
\end{align}

We can apply the chase procedure to verify if the fact $\FP{Spare}(c)$ follows from \KB.
In this case, the \emph{restricted chase} produces a universal model set with two models for \KB; namely, $\{\FP{Engine}(d), \FP{Spare}(d)\}$
and $\{\FP{Engine}(d), \FP{IsIn}(d, \SF{}{v}(d)), \FP{Bike}(\SF{}{v}(d)), \FP{Has}(\SF{}{v}(d), d)\}$, where \SF{}{v}(d)) is a fresh term introduced to satisfy the existential quantifier in \eqref{main1}.
Since the second model does not contain $\FP{Spare}(d)$, this fact is not entailed by \KB.
\end{example}

Since boolean conjunctive query entailment is undecidable \cite{DBLP:conf/icalp/BeeriV81}, the chase may not always terminate.
Even worse, we cannot decide if the chase terminates on a particular KB or if a rule set \R is \emph{acyclic} \cite{AllInstanceTerminationUndecidable,AnatomyOfTheChase}; that is, if the chase terminates for every KB with \R.
We can still verify rule set termination in practice using \emph{acyclicity notions}; that is, sufficient conditions that imply termination
\cite{DBLP:journals/tcs/FaginKMP05,GeneralizedSchemaMappings,DBLP:conf/ijcai/KrotzschR11,MFA,RMFA,DBLP:conf/ecai/BagetGMR14,DBLP:journals/tplp/KarimiZY21}.
However, if an acyclicity notion is not able to classify a rule set \R as terminating, we never know if this notion is just not ``general enough'' or if the rule set is indeed non-terminating.

To address this issue, we study \emph{cyclicity notions}, which imply non-termination.
As a long-term motivation, these approaches can also help to fix potential modelling mistakes.
To the best of our knowledge, only one such notion has been proposed for the restricted chase variant, namely \emph{Restricted Model Faithful Cyclicity (\RMFC)}.
Alas, many non-terminating rule sets with disjunctions are not classified as such by this notion \cite{RMFA}.
Worse still, the correctness proof of \RMFC does not hold in its presented form (see Section~\ref{section:related-work}).
Recently, \citeauthor{DMFA} have also proposed \emph{Disjunctive Model Faithful Cyclicity (\DMFC)} for the skolem chase.
Note however that a cyclicity notion for the skolem chase is not directly a valid condition for restricted non-termination since the former variant terminates less often than the latter.

In this paper, our overarching goal is to improve our understanding of existing cyclicity notions such as \RMFC and \DMFC.
We reconsider the underlying ideas and dismantle them into an extensible framework.
We provide examples to explain how these notions work and clarify why certain technical details are necessary.
As more tangible contributions,
\FirstItem we come up with novel cyclicity notions named \emph{restricted prefix cyclicity (\RPC)} and \emph{deterministic \RPC (\DRPC)},
and \SecondItem we empirically evaluate the generality of these two.

To this aim, the key points of the sections are as
follows:\begin{paper}
  \footnote{More proof details are online \cite{kr2023tr}.}
\end{paper}
\begin{tr}
  \footnote{This report features additional proof details in the appendix.}
\end{tr}

\begin{enumerate}
  \item[S3.] \emph{Cyclicity sequences} that guarantee \emph{never-termination} by making use of \emph{g-unblockability}.
  \item[S4.] Checkable conditions that ensure g-unblockability.
  \item[S5.] \emph{Cyclicity prefixes} that allow cyclicity sequences.
  \item[S6.] The notion (D)RPC that guarantees a cyclicity prefix.
  \item[S7.] Detailed relations to \DMFC and \RMFC in particular.
  \item[S8.] Empirical evaluation of the generality of (D)RPC.
\end{enumerate}

\section{Preliminaries}\label{section:prelim}

We define \Preds, \Funs, \Constants, and \Vars to be mutually disjoint, countably infinite sets of predicates, function symbols, constants, and variables, respectively.
Every $s \in \Preds \cup \Funs$ is associated with an \emph{arity} $\Arity(s) \geq 1$.
The set $\Terms \supseteq \C \cup \Vars$ contains $f(t_1, \dots, t_n)$ for every $n \geq 1$, every $f \in \Funs$ with $\Arity(f) = n$, and every $t_1, \ldots, t_n \in \Terms$.
For some $\FormatEntitySet{X} \in \{\Preds, \Funs, \Constants, \Vars, \Terms\}$ and an expression $\Expression$, we write $\EI{\FormatEntitySet{X}}{\Expression}$ to denote the set of all elements of $\FormatEntitySet{X}$ that syntactically occur in $\Expression$.

A term $t \notin \Variables \cup \Constants$ is \emph{functional}.
For a term $t$; let $\Depth(t) = 1$ if $t$ is not functional, and $\Depth(t) = 1 + \Max(\Depth(s_1), \ldots, \Depth(s_n))$ if $t$ is of the form $f(s_1, \ldots, s_n)$.
We write lists $t_1, \ldots, t_n$ of terms as $\Vt$ and often treat these as sets.
A term $s$ is a \emph{subterm} of another term $t$ if $t = s$, or $t$ is of the form $f(\Vs)$ and $s$ is a subterm of some term in $\Vs$.
For a term $t$, let $\Subterms(t)$ be the set of all subterms of $t$.
A term is \emph{cyclic} if it has a subterm of the form $f(\Vs)$ with $f \in \EI{\Funs}{\Vs}$.

An \emph{atom} is an expression of the form $\FP{P}(\vec{t})$ with \FP{P} a predicate and $\vec{t}$ a list of terms such that $\Arity(\FP{P}) = \vert \Vt \vert$.
A \emph{fact} is a variable-free atom.
For a formula \Formula, we write $\Formula(\Vx)$ to denote that \Vx is the set of all free variables that occur in \Formula.

\begin{definition}
A \emph{(disjunctive existential) rule} is a constant- and function-free first-order formula of the form
\begin{align}
\forall \vec{w}, \vec{x}.[\Body(\vec{w}, \vec{x}) \to \bigvee\nolimits_{i=1}^{n} \exists \vec{y}_{i}. \Head_{i}(\vec{x}_{i}, \vec{y}_{i})] \label{rule}
\end{align}
where $n \geq 1$; $\vec{w}, \vec{x}, \vec{y}_{1}, \dots,$ and $\vec{y}_{n}$ are pairwise disjoint lists of variables; $\bigcup_{i=1}^{n} \vec{x}_{i} = \vec{x}$; and $\Body, \Head_{1}, \dots,$ and $\Head_{n}$ are non-empty conjunctions of atoms.
\end{definition}

We omit universal quantifiers when writing rules and often treat conjunctions as sets.
The \emph{frontier} of a rule \Rule such as \eqref{rule} is the variable set $\Frontier(\Rule) = \Vx$.
Moreover, let $\BF(\Rule) = \Body$, let $\HF_{i}(\Rule) = \Head_{i}$ for every $1 \leq i \leq n$, and let $\BranchFactor(\Rule) = n$.
The rule \Rule is \emph{deterministic} if $n = 1$, \emph{generating} if $\vec{y}_{i}$ is non-empty for some $1 \leq i \leq n$, and \emph{datalog} if it is deterministic and non-generating.

A \emph{(ground) substitution} is partial function from variables to ground terms; that is, to variable-free terms.
We write $[x_1 / t_1, \dots, x_n / t_n]$ to denote the substitution mapping $x_{1}, \dots, x_{n}$ to $t_{1}, \dots, t_{n}$, respectively.
For an expression $\Expression$ and a substitution \Subs, let $\Expression\Subs$ be the expression that results from $\Expression$ by uniformly replacing every syntactic occurrence of every variable $x$ by $\Subs(x)$ if the latter is defined.

For a rule \Rule such as \eqref{rule}, let $\Subs_\Rule$ be the substitution mapping $y$ to $\SF{\Rule}{i, y}(\Vx)$ for every $1 \leq i \leq n$ and every $y \in \Vy_i$ with \SF{\Rule}{i, y} a fresh function symbol unique for $\Tuple{\Rule, i, y}$.
If $y$ uniquely identifies the tuple $\Tuple{\Rule, i, y}$, we also write $\SF{}{y}(\Vx)$. (This is the case in all our examples.)
The \emph{skolemization} $\Skolemise(\Rule)$ of \Rule is the expression $\Body \to (\bigvee_{i=1}^{n} \Head_i)\Subs_\Rule$.
For every $1 \leq i \leq n$, let $\HF_i(\Skolemise(\Rule)) = \Head_i\Subs_\Rule$.

A \emph{trigger} \Trigger is a pair $\Tuple{\Rule, \Subs}$ with \Rule a rule and \Subs a substitution with domain $\EI{\Variables}{\BF(\Rule)}$.
A trigger is \emph{loaded} for a fact set \F if $\BF(\Rule)\Subs \subseteq \F$.
It is \emph{obsolete} for \F if there is a substitution \SubsAux that extends \Subs such that $\HF_{i}(\Rule)\SubsAux \subseteq \F$ for some $1 \leq i \leq \BranchFactor(\Rule)$.
Let $\Output_{i}(\Trigger) = \HF_{i}(\Skolemise(\Rule))\Subs$ for every $1 \leq i \leq \BranchFactor(\Rule)$, and $\Output(\Trigger) = \{\Output_{i}(\Trigger) \mid 1 \leq i \leq \BranchFactor(\Rule)\}$.

Consider a rule set \R.
An \emph{\R-term} is a term defined using the function symbols that occur in $\Skolemise(\R)$, some constants, and some variables.
A substitution is an \emph{\R-substitution} if its range is a set of \R-terms.
An \emph{\R-trigger} is a trigger with a rule from \R and an \R-substitution.

A fact set \F \emph{satisfies} a rule \Rule if all triggers with \Rule are not loaded or obsolete for \F.
A \emph{knowledge base (KB)} is a pair $\KB = \Tuple{\R, \D}$ of a rule set \R and a \emph{database} \D; that is, a function-free fact set.
The \emph{restricted chase} on input \KB exhaustively applies the outputs of triggers that are loaded and not obsolete in a tree with root \D branching on disjunctions.


\subsubsection{Restricted Chase}
We present a variant of the disjunctive chase \cite{DBLP:journals/tods/BourhisMMP16} where the application of rules is \emph{restricted}; that is, rules are only applicable if their heads are not obsolete with respect to previously derived facts.
Moreover, we impose an order of rule applications that prioritises the application of (triggers with) datalog rules.

\begin{definition}\label{def:chase-tree}
A \emph{chase tree} $\CT = \Tuple{V, E, \LF, \LT}$ for a KB $\Tuple{\R, \D}$ is a directed tree where $V$ is a set of vertices, $E$ is a set of edges, \LF is a function mapping vertices to fact sets, and \LT is a function mapping every non-root vertex to a trigger.
Moreover, the following hold:
\begin{enumerate}
\item For the root $r \in V$ of \CT, we have that $\LF(r) = \D$.
\item For every non-leaf vertex $v \in V$; there is an \R-trigger $\Trigger = \Tuple{\Rule, \Subs}$ that is loaded and not obsolete for $\LF(v)$ such that
\FirstItem the set $\LF(v)$ satisfies all datalog rules in \R if the rule in \Trigger is not datalog,
\SecondItem $v$ has exactly $n = \BranchFactor(\Rule)$ children $c_{1}, \dots, c_{n}$ (via $E$) with
$\LF(c_{i}) = \LF(v) \cup \Output_{i}(\Trigger)$
and
$\LT(c_{i}) = \Trigger$ for each $1 \leq i \leq n$.


\item For every leaf vertex $v \in V$, the set $\LF(v)$ satisfies all of the rules in \R.
For every \R-trigger \Trigger that is loaded for $\LF(v)$ for some $v \in V$, there is a $k \geq 0$ such that \Trigger is obsolete for $\LF(u)$ for each $u \in V$ reachable from $v$ by a path of length $k$.
That is, fairness.
\end{enumerate}
\end{definition}

We refer to $\LF(v)$ and $\LT(v)$ for some $v \in V$ as the \emph{fact-} and \emph{trigger-label} of $v$, respectively.
Informally, we say that a trigger (resp. a rule) is applied in a chase tree to signify that some vertex in the tree is labelled with this trigger (resp. a trigger with this rule).

\begin{example}
\label{example:chase-tree}
The KB from Example~\ref{exp:main} only admits one chase tree, which is depicted in
 Figure~\ref{fig:ctExpMain}.
 \begin{figure}
 \centering
 \begin{tikzpicture}
 \node(A) at (-2,0) {$\FP{Engine}(d)$};

 \node(B) at (2,0) {$\dots, \FP{Spare}(d) \Tuple{\eqref{main1}, [x / d]}$};

 \node(C) at (0,-0.8) {$\dots, \FP{IsIn}(d, \SF{}{v}(d)), \FP{Bike}(\SF{}{v}(d)) \Tuple{\eqref{main1}, [x / d]}$};

 \node(D) at (0,-1.8) {$\dots, \FP{Has}(\SF{}{v}(d), d) \Tuple{\eqref{main3}, [x / d, y / \SF{}{v}(d)]}$};

 \path [thick, ->](A) edge (B);
 \path [thick, ->](A) edge (C);
 \path [thick, ->](C) edge (D);
 \end{tikzpicture}
 \caption{Chase Tree for Example~\ref{exp:main}}
 \label{fig:ctExpMain}
 \end{figure}
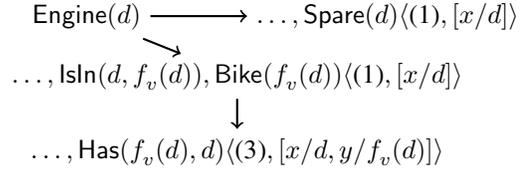
\end{example}

A \emph{branch} of a chase tree $\CT = \Tuple{V, E, \LF, \LT}$ is a maximal path in \CT starting at the root; that is, a vertex sequence $B = v_{1}, v_{2}, \dots$ where $v_1$ is the root, $\Tuple{v_{i}, v_{i+1}} \in E$ for every $1 \leq i < \vert B \vert$, and the last element of $B$ is a leaf if $B$ is finite.
The \emph{result} of $\CT$ is the set $\{\bigcup_{v \in B} \LF(v) \mid B \text{ a branch of } \CT \}$; chase results can be used to solve query entailment:

\begin{proposition}
\label{proposition:chase-soundness}
Consider the result $\mathfrak{R}$ of some (arbitrarily chosen) chase tree of a \KB.
Then, \KB entails a query $\Query = \exists \Vy . \Body$ iff $\F \models \Query$ for every $\F \in \mathfrak{R}$ iff for every $\F \in \mathfrak{R}$ there is a substitution \Subs with $\Body\Subs \subseteq \F$.
\end{proposition}

Therefore, it is interesting to know if a rule set admits finite chase trees.
A rule set \R \emph{terminates} if every chase tree of every KB with \R is finite, it \emph{sometimes-terminates} if every KB with \R admits a finite chase tree, 
and it \emph{never-terminates} if some KB with \R has no finite chase trees.

We use skolem function names to backtrack the facts along which a term $t$ appears in the chase; that is, the \emph{birth facts} of $t$.
For a constant $c$, let $\BirthFacts{\R}{c} = \emptyset$; for a rule set \R and an \R-term $t$ of the form $\SF{\Rule}{i, v}(\Vs)$, let $\BirthFacts{\R}{t} = \Output_i(\Tuple{\Rule, \Subs}) \cup \bigcup_{s \in \Vs} \BirthFacts{\R}{s}$ where $\Subs$ is a substitution with $\Frontier(\Rule)\Subs = \Vs$.
For a trigger $\Tuple{\RuleAux, \SubsAux}$, let $\BirthFacts{\R}{\Tuple{\RuleAux, \SubsAux}} = \bigcup\nolimits_{x \in \Frontier(\RuleAux)}\BirthFacts{\R}{\SubsAux(x)}$.
The \emph{term-skeleton} $\Skeleton{\R}{\Tuple{\RuleAux, \SubsAux}}$ of $\Tuple{\RuleAux, \SubsAux}$ consists of $\EI{\Terms}{\BirthFacts{\R}{\Tuple{\RuleAux, \SubsAux}}}$ and
every constant $c$ with $c = \SubsAux(x)$ for an $x \in \Frontier(\RuleAux)$.

\begin{example}\label{exp:birthAndSkeleton}
  Let $\R = \{\eqref{main1}, \eqref{main2}\}$
  and $\Trigger = \Tuple{\eqref{main2}, [x / \SF{}{v}(d)]}$.
  We have $\BirthFacts{\R}{\Trigger} = \{ \FP{IsIn}(d, \SF{}{v}(d)), \FP{Bike}(\SF{}{v}(d)) \}$
  and $\Skeleton{\R}{\Trigger} = \{ d, \SF{}{v}(d) \}$. (Note again that $\SF{}{v} = \SF{\eqref{main1}}{1, v}$.)
\end{example}

\section{Cyclicity Sequences}

In this section, we introduce the notion of a \emph{cyclicity sequence} \TriggerSeq for a KB $\KB = \Tuple{\R, \D}$ (see Definition~\ref{definition:cyclicity-sequence}) and show that its existence implies that \KB only admits infinite chase trees (see Theorem~\ref{theorem:cyclicity-sequence}).
Intuitively, \TriggerSeq is an infinite sequence of \R-triggers that are applied in some (infinite) branch of every chase tree of \KB.
To identify these branches, we consider the following notion from \cite{DMFA}.

\begin{definition}
\label{definition:head-choice}
A \emph{head-choice} for a rule set \R is a function $\HeadChoice$ that maps every rule $\Rule \in \R$ to an element of $\{1, \ldots, \BranchFactor(\Rule)\}$.
For a rule $\Rule \in \R$ and a trigger \Trigger with \Rule, we write $\HF_{\HeadChoice}(\Rule)$ and $\Output_{\HeadChoice}(\Trigger)$ instead of $\HF_{\HeadChoice(\Rule)}(\Rule)$ and $\Output_{\HeadChoice(\Rule)}(\Trigger)$, respectively.
For a chase tree $\CT = \Tuple{V, E, \LF, \LT}$ of a KB with \R, we define $\Branch(\CT, \HeadChoice) = v_1, v_2, \ldots$ as the branch of \CT such that $\LF(v_{i+1}) = \LF(v_{i}) \cup \Output_\HeadChoice(\LT(v_{i+1}))$ for every $1 \leq i < \vert \Branch(\CT, \HeadChoice) \vert$;
note that every branch starts with the root.
\end{definition}

Cyclicity sequences must satisfy three requirements (see Definition~\ref{definition:cyclicity-sequence}); let's start with the first two:

\begin{definition}
Consider a KB $\KB = \Tuple{\R, \D}$, a head-choice \HeadChoice for \R, and a sequence $\TriggerSeq = \Trigger_1, \Trigger_2, \ldots$ of \R-triggers.
\begin{itemize}
\item Let $\F_0(\KB, \HeadChoice, \TriggerSeq) = \D$, and let $\F_{i+1}(\KB, \HeadChoice, \TriggerSeq) = \F_{i}(\KB, \HeadChoice, \TriggerSeq) \cup \Output_{\HeadChoice}(\Trigger_{i+1})$ for every $0 \leq i < \vert \TriggerSeq \vert$.
\item The sequence \TriggerSeq is \emph{loaded} for \KB and \HeadChoice if $\Trigger_{i+1}$ is loaded for $\F_{i}(\KB, \HeadChoice, \TriggerSeq)$ for every $0 \leq i < \vert \TriggerSeq \vert$.
\item The sequence \TriggerSeq is \emph{growing} for \KB and \HeadChoice if, for every $0 \leq i < \vert \TriggerSeq \vert$, there is some $j > i$ and a term that occurs in $\F_{j}(\KB, \HeadChoice, \TriggerSeq)$ but not in $\F_{i}(\KB, \HeadChoice, \TriggerSeq)$.
Note that \TriggerSeq is infinite if this requirement is satisfied.
\end{itemize}
\end{definition}

For some variants of the chase, existence of a loaded and growing sequence of \R-triggers for a KB and a head-choice may suffice to witness non-termination.
This is not the case for the restricted chase:

\begin{example}
\label{example:loaded-growing-not-enough}
Consider the KB $\KB = \Tuple{\R, \D}$ from Example~\ref{exp:main}, and the head-choice $\HeadChoice$ that maps every rule in $\R$ to $1$.
Moreover, consider the infinite sequence $\TriggerSeq = \Trigger_1, \Trigger_2, \ldots$ of triggers such that:
\begin{itemize}
\item Let $t_0 = d$, let  $t_i = \SF{}{v}(t_{i-1})$ for every odd $i \geq 1$, and let $t_i = \SF{}{w}(t_{i-1})$ for every even $i \geq 1$.
\item For every $i \geq 1$; let $\Trigger_i = \{\eqref{main1}, [x / t_{i-1}]\}$ if $i$ is odd, and $\Trigger_i = \{\eqref{main2}, [x / t_{i-1}]\}$ otherwise.
\end{itemize}
The (infinite) sequence $\TriggerSeq$ is loaded and growing; however, the KB $\KB$ terminates!
Note that this KB only admits one chase tree, which is finite (see Example~\ref{example:chase-tree}).
\end{example}

The issue in the previous example is that the second trigger of \TriggerSeq cannot be applied in any chase tree of $\KB$;
this is because we must apply datalog triggers with \eqref{main3} before we apply triggers with \eqref{main2}.
To address this issue, we introduce a third requirement for cyclicity sequences (see Definition~\ref{definition:cyclicity-sequence}):

\begin{definition}
\label{definition:unblockability}
Consider a KB $\KB = \Tuple{\R, \D}$, a head-choice \HeadChoice for \R, and a sequence of \R-triggers $\TriggerSeq = \Trigger_1, \Trigger_2, \ldots$
\begin{itemize}
\item An \R-trigger \Trigger is \emph{g-unblockable}\footnote{The ``g-'' prefix stands for ``general-''; we introduce more specific unblockability notions in the following section.} for \KB and \HeadChoice if, for every chase tree $\CT = \Tuple{V, E, \LF, \LT}$ and every $v$ in $\Branch(\CT, \HeadChoice)$ such that $\Trigger$ is loaded for $\LF(v)$, there is some $u$ in $\Branch(\CT, \HeadChoice)$ with $\Output_{\HeadChoice}(\Trigger) \subseteq \LF(u)$.
Intuitively, if \Trigger is loaded for a vertex in the \HeadChoice-branch of a chase tree \CT of \KB, then its output according to \HeadChoice eventually appears in this branch.
  \item If every trigger in \TriggerSeq is g-unblockable, then this sequence is \emph{g-unblockable}
        for \KB and \HeadChoice.
\end{itemize}
\end{definition}

The sequence \TriggerSeq in Example~\ref{example:loaded-growing-not-enough} is not g-unblockable because its second trigger does not satisfy this property.
However, \TriggerSeq is g-unblockable for a slightly different input KB:

\begin{example}
  Consider the rule set $\R = \{\eqref{main1}, \eqref{main2}\}$;
  and the database $\D$, the head-choice $\HeadChoice$, and the sequence $\TriggerSeq$
  from Example~\ref{example:loaded-growing-not-enough}.
Since $\R$ contains neither \eqref{main3} nor \eqref{main4}, the sequence $\TriggerSeq$ is g-unblockable for $\Tuple{\R, \D}$ and $\HeadChoice$.
Note that the KB $\Tuple{\R, \D}$ only admits one chase tree, which is infinite, and hence, \R is never-terminating.
\end{example}

We are ready to define cyclicity sequences:

\begin{definition}
\label{definition:cyclicity-sequence}
A sequence $\TriggerSeq = \Trigger_1, \Trigger_2, \ldots$ of \R-triggers is a \emph{cyclicity sequence} of a KB $\KB = \Tuple{\R, \D}$ and a head-choice \HeadChoice if \TriggerSeq is (infinite,) loaded, growing, and g-unblockable.
Note that \TriggerSeq is infinite if it is growing.
\end{definition}


\begin{theorem}
\label{theorem:cyclicity-sequence}
A rule set \R never-terminates if there is a cyclicity sequence for a KB with \R and some head-choice.
\end{theorem}
\begin{proof}
Assume that there is a cyclicity sequence $\TriggerSeq = \Trigger_1, \Trigger_2, \ldots$ of some KB such as $\KB = \Tuple{\R, \D}$ and some head-choice \HeadChoice, and consider some chase tree $\CT = \Tuple{V, E, \LF, \LT}$ of $\Tuple{\R, D}$ and the sequence $\Branch(\CT, \HeadChoice) = v_1, v_2, \ldots$
To show the theorem, we prove that the fact set $\F(\CT, \HeadChoice) = \bigcup_{i \geq 0} \LF(v_i)$ is infinite, which implies that both $\Branch(\CT, \HeadChoice)$ and \CT are infinite.
This holds by infinity of $\bigcup_{i \geq 0} \F_i(\KB, \HeadChoice, \TriggerSeq)$ (since \TriggerSeq is growing) and $\F_i(\KB, \HeadChoice, \TriggerSeq) \subseteq \F(\CT, \HeadChoice)$ for every $i \geq 0$ (which follows by induction since \TriggerSeq is loaded and g-unblockable.)
\end{proof}

\section{Infinite Unblockable Sequences}
\label{section:unblockability}

Theorem~\ref{theorem:cyclicity-sequence} provides a blueprint to show never-termination of a rule set \R: One ``simply'' has to show that \R admits a cyclicity sequence (see Section~\ref{section:cyclicity-notions}).
To show that such sequences exist, we have developed techniques to find infinite sequences of triggers that are g-unblockable.
This is a rather challenging task; note that we cannot even decide if a single trigger is g-unblockable (by reduction from fact entailment \cite{DBLP:conf/icalp/BeeriV81}):
\begin{theorem}
\label{theorem:g-unblockability-undecidable}
The problem of deciding if a trigger is g-unblockable for a KB and a head-choice is undecidable.
\end{theorem}

In this section, we first discuss ways to detect if a trigger is g-unblockable in some cases.
Then, we devise strategies to show that some infinite sequences of triggers are g-unblockable, i.e. that unblockability ``propagates''.

\subsubsection{Detecting Unblockability}

To detect if a trigger \Trigger is g-unblockable,
we make use of chase over-approximations before the application of \Trigger:

\begin{definition}\label{def:semantic-over-approximation}
Consider a rule set \R, a head-choice \HeadChoice, and some \R-trigger $\Trigger = \Tuple{\Rule, \Subs}$.
A fact set \F is an \emph{over-approximation} of \R and \HeadChoice before \Trigger
  if there is a function $h$ over the set of terms such that \FirstItem $h(\Subs(x)) = \Subs(x)$ for each $x \in \Frontier(\Rule)$ and, \SecondItem for every $u \in \Branch(\CT, \HeadChoice)$ in every chase tree $\CT = \Tuple{V, E, \LF, \LT}$ of every KB \Tuple{\R, \D} with $\Output_\HeadChoice(\Trigger) \nsubseteq \LF(u)$, we have that $h(\LF(u)) \subseteq \F$.
\end{definition}

Intuitively, a chase over-approximation such as \F above for \R and \HeadChoice before \Trigger is some sort of ``upper-bound'' of all of the facts that can possibly occur in the label of a vertex in the \HeadChoice-branch if this label does not include the output of \Trigger.
If \Trigger is not obsolete for \F, its output eventually appears in the \HeadChoice-branch of the chase:

\begin{lemma}
\label{lemma:over-approximation-hc}
If \Trigger is not obsolete for some over-approximation of a rule set \R and a head-choice \HeadChoice before \Trigger, then \Trigger is g-unblockable for \HeadChoice and every KB with \R.
\end{lemma}
\begin{proof}
To show the contrapositive, we assume that $\Trigger = \Tuple{\Rule, \Subs}$ is not g-unblockable for some $\Tuple{\R, \D}$ and \HeadChoice.
Then, there exists a chase tree $\CT = \Tuple{V, E, \LF, \LT}$ of $\Tuple{\R, \D}$ and \HeadChoice such that \FirstItem \Trigger is loaded for some $v \in \Branch(\CT, \HeadChoice)$ and \SecondItem $\Output_{\HeadChoice}(\Trigger) \nsubseteq \LF(u)$ for every $u \in \Branch(\CT, \HeadChoice)$.
Moreover, \Trigger is obsolete for $\LF(w)$ for some $w \in \Branch(\CT, \HeadChoice)$ by Definition~\ref{def:chase-tree}.
By Definition~\ref{def:semantic-over-approximation}, we find $h(\LF(w)) \subseteq \F$ for every over-approximation \F of \R and \HeadChoice before \Trigger with term mapping $h$.
Then, $\Tuple{\Rule, h \circ \Subs}$ is obsolete for \F; hence so is $\Trigger$ since $h \circ \Subs$ and $\Subs$ agree on 
all frontier variables of $\Rule$.
\end{proof}



Lemma~\ref{lemma:over-approximation-hc}
provides a strategy to detect g-unblockability for a given trigger \Trigger: Compute some over-approximation \F and then check if \Trigger is obsolete for \F.
We consider two alternative ways of computing these over-approximations:
\begin{definition}
\label{definition:over-approximation}
For a trigger \Trigger, let $h_\Trigger^\star$ be the function over the set of terms that maps every $t \in \Skeleton{\R}{\Trigger}$ to itself and every other term to the special constant $\star$.
Moreover, let $h_\Trigger^\textit{uc}$ be another such function that maps every functional term $f(\vec{t}) \notin \Skeleton{\R}{\Trigger}$ to a fresh constant $c_f$, every term in \Skeleton{\R}{\Trigger} and every constant of the form $c_{f}$ to itself, and every other term to $\star$.


For a rule set \R, a head-choice \HeadChoice, some \R-trigger $\Trigger = \Tuple{\Rule, \Subs}$,
and some $h \in \{h_\Trigger^\star, h_\Trigger^\textit{uc}\}$; let $\mathcal{O}(\R, \HeadChoice, \Trigger, h)$ be the minimal fact set that
\begin{enumerate}
\item contains every fact that can be defined using a predicate occurring in \R and constants in $\Skeleton{\R}{\Trigger} \cup \{ \star \}$, \label{condition:crit-inst}
\item includes $\BirthFacts{\R}{\Trigger}$, and \label{condition:birth-facts}
  \item includes $h(\Output_\HeadChoice(\Trigger'))$ for every \R-trigger $\Trigger'$ such that $\Trigger'$ is loaded for $\mathcal{O}(\R, \HeadChoice, \Trigger, h)$ and $\Output_\HeadChoice(\Trigger') \neq \Output_\HeadChoice(\Trigger)$.
\end{enumerate}
Moreover, we define $\mathcal{O}(\R, \Trigger, h)$ as the minimal fact set that satisfies \eqref{condition:crit-inst}
and \eqref{condition:birth-facts} above, and includes $h(\bigcup \Output(\Trigger'))$
for every \R-trigger $\Trigger' = \Tuple{\RuleAux, \SubsAux}$ such that $\Trigger'$ is loaded for $\mathcal{O}(\R, \Trigger, h)$
and if $\RuleAux = \Rule$, then $\Output_{i}(\Trigger') \neq \Output_{i}(\Trigger)$ for some $1 \leq i \leq \BranchFactor(\Rule)$.
\end{definition}

Intuitively, $\mathcal{O}(\R, \Trigger, h)$ views outputs as if disjunctions were replaced by conjunctions in rules.

\begin{example}
For the rule set $\R = \{\eqref{main1}, \eqref{main2}\}$, the head-choice $\HeadChoice_1 : \R \to \{1\}$, some constant $d$, and the trigger $\Trigger = \Tuple{\eqref{main2}, [x / \SF{}{v}(d)]}$; the set $\mathcal{O}(\R, \HeadChoice_{1}, \Trigger, h^\textit{uc}_\Trigger)$ equals
\begin{align*}
&\{ \FP{Engine}(s), \FP{Bike}(s), \FP{Spare}(s), \FP{IsIn}(s, t) \mid s, t \in \{\star, d\}\} \cup~ \\
&\BirthFacts{\R}{\Trigger} \cup \{ \FP{IsIn}(\star, c_{v}), \FP{Bike}(c_{v}), \FP{Has}(\star, c_{w}), \\
&\FP{Has}(d, c_{w}), \FP{Engine}(c_{w}), \FP{IsIn}(c_{w}, c_{v}), \FP{Has}(c_{v}, c_{w}) \}.
\end{align*}
In the above, we write $c_v$ and $c_w$ to refer to the fresh constants unique for $\SF{}{v}$ and $\SF{}{w}$, respectively.
Note that $\mathcal{O}(\R, \HeadChoice_{1}, \Trigger, h^\textit{uc}_\Trigger)$ does not contain $\FP{IsIn}(d, c_{v})$ since $h^{\textit{uc}}_{\Trigger}$ maps $\SF{}{v}(d)$ to itself,
or $\FP{Has}(\SF{}{v}(d), c_{w})$ since this fact is in the output of a trigger excluded by Item~3 in Definition~\ref{definition:over-approximation}.
The set $\mathcal{O}(\R, \Trigger, h^\textit{uc}_\Trigger)$ includes $\mathcal{O}(\R, \HeadChoice_{1}, \Trigger, h^\textit{uc}_\Trigger)$ and additionally contains $\FP{Spare}(c_{w})$.
If we replace all occurrences of $c_v$ and $c_w$ in $\mathcal{O}(\R, \HeadChoice_{1}, \Trigger, h^\textit{uc}_\Trigger)$ (resp. $\mathcal{O}(\R, \Trigger, h^\textit{uc}_\Trigger))$ with $\star$, we obtain $\mathcal{O}(\R, \HeadChoice_{1}, \Trigger, h^\star_\Trigger)$ (resp. $\mathcal{O}(\R, \Trigger, h^\star_\Trigger)$).
\end{example}



\begin{lemma}
\label{lemma:star-uc-over-approximations}
For a rule set \R, a head-choice \HeadChoice, an \R-trigger $\Trigger = \Tuple{\Rule, \Subs}$, and some $h \in \{h_\Trigger^\star, h_\Trigger^\textit{uc}\}$;
$\mathcal{O}(\R, \HeadChoice, \Trigger, h)$ and $\mathcal{O}(\R, \Trigger, h)$ are over-approximations of \R and \HeadChoice before \Trigger.
\end{lemma}
\begin{proof}

  We show that $\mathcal{O}(\R, \HeadChoice, \Trigger, h)$
  is an over-approximation of \R and \HeadChoice before \Trigger.
  The first condition of Definition~\ref{def:semantic-over-approximation} holds for $h$ by
  Definition~\ref{definition:over-approximation}.
  For every $u \in \Branch(\CT, \HeadChoice)$ 
  in every chase tree \CT of every KB $\Tuple{\R, \D}$,
  we can verify the second condition via induction on 
  the path $u_{1}, \dots, u_{n}$ in \CT with the root $u_{1}$ and $u_{n} = u$
  assuming that
  $\Output_{\HeadChoice}(\Trigger) \nsubseteq \LF(u)$.
  The base case holds since the facts in $h(\D)$ are contained in the facts defined 
  by \eqref{condition:crit-inst} in Definition~\ref{definition:over-approximation}.
  For the induction step it is important to realize that for each $2 \leq i \leq n$,
  $\Output_{\HeadChoice}(\LT(u_{i})) \neq \Output_{\HeadChoice}(\Trigger)$ holds by $\Output_{\HeadChoice}(\Trigger) \nsubseteq \LF(u)$.

  For the second part of the claim, $\mathcal{O}(\R, \Trigger, h)$
  is an over-approximation of \R and \HeadChoice before \Trigger
  since $\mathcal{O}(\R, \HeadChoice, \Trigger, h) \subseteq \mathcal{O}(\R, \Trigger, h)$ by Definition~\ref{definition:over-approximation}.
\end{proof}

Using the over-approximations from Definition~\ref{definition:over-approximation}, we define two different types of unblockability:
\begin{definition}
\label{definition:decidable-unblockability}
Consider a rule set \R and an \R-trigger $\Trigger$.
Then, \Trigger is \emph{$\star$-unblockable} for \R if it features a datalog rule or it is not obsolete for $\mathcal{O}(\R, \Trigger, h^\star_\Trigger)$.
Moreover, it is \emph{$\textit{uc}$-unblockable} for \R and some \HeadChoice if it features a datalog rule or it is not obsolete for $\mathcal{O}(\R, \HeadChoice, \Trigger, h^\textit{uc}_\Trigger)$.
\end{definition}

\begin{example}\label{exp:mainUnblk}
Consider the rule set $\R = \{\eqref{main1}, \eqref{main2}\}$ and the head-choice $\HeadChoice_1$ mapping all rules to 1.
The trigger $\Trigger = \Tuple{\eqref{main2}, [x / \SF{}{v}(d)]}$ is \textit{uc}-unblockable for \R and $\HeadChoice_{1}$; it is not for $\R' = \R \cup \{\eqref{main3}\}$ and $\HeadChoice_{1}$.
Note that $\mathcal{O}(\R', \HeadChoice_{1}, \Trigger, h^\textit{uc}_\Trigger)$ includes $\mathcal{O}(\R, \HeadChoice_{1}, \Trigger, h^\textit{uc}_\Trigger) \cup \{\FP{Has}(\SF{}{v}(d), d)\}$ (among other facts).
Therefore, \Trigger is obsolete for $\mathcal{O}(\R', \HeadChoice_{1}, \Trigger, h^\textit{uc}_\Trigger)$ but not for $\mathcal{O}(\R, \HeadChoice_{1}, \Trigger, h^\textit{uc}_\Trigger)$.
\end{example}

By Theorem~\ref{theorem:g-unblockability-undecidable}, we cannot decide if a trigger \Trigger is g-unblockable.
However, we can effectively check $\star$- or $\textit{uc}$-unblockability; both properties imply g-unblockability:

\begin{lemma}
\label{lemma:deciding-unblockability}
If an \R-trigger \Trigger is $\star$-unblockable for a rule set \R, it is \textit{uc}-unblockable for \R and every head-choice.
If \Trigger is \textit{uc}-unblockable for \R and some head-choice \HeadChoice, then it is g-unblockable for \R and \HeadChoice.
\end{lemma}
\begin{proof}
  The first implication holds since $\mathcal{O}(\R, \Trigger, h^\star_\Trigger)$ includes $h(\mathcal{O}(\R, \HeadChoice, \Trigger, h^\textit{uc}_\Trigger))$
  with $h$ the function that maps every fresh constant in the range of $h^\textit{uc}_\Trigger$ to $\star$.
  For the second, note that a trigger \Trigger with a datalog rule is g-unblockable.
  For the non-datalog case, we can apply Lemmas~\ref{lemma:over-approximation-hc} and \ref{lemma:star-uc-over-approximations}.
\end{proof}
By Lemma~\ref{lemma:deciding-unblockability}, a trigger is \textit{g}-unblockable if it is $\star$-unblockable; we can also prove this directly with Lemmas~\ref{lemma:over-approximation-hc} and \ref{lemma:star-uc-over-approximations}.
By Lemma~\ref{lemma:deciding-unblockability}, \textit{uc}-unblockability is more general than $\star$-unblockability; the other direction does not hold:
\begin{example}
  Consider the following rule set \R:
  \begin{align}
    &R(x, y) \to \exists u. R(y, u) \label{ucButNotStarUnblk1}\\
    &R(x, y) \to \exists v. S(y, v)\\
    &R(x, y) \to \exists w. T(y, w)\\
    &S(x, y) \land T(x, y) \to R(x, y)
  \end{align}
  The trigger $\Tuple{\eqref{ucButNotStarUnblk1}, [x / c_{y}, y / \SF{}{u}(c_{y})]}$
  is \textit{uc}-unblockable but not $\star$-unblockable:
  We find $S(\SF{}{u}(c_{y}), \star), T(\SF{}{u}(c_{y}), \star)$ and therefore $R(\SF{}{u}(c_{y}), \star)$
  in $\mathcal{O}(\R, \Trigger, h^\star_\Trigger)$.
  On the other hand, we only find
  $S(\SF{}{u}(c_{y}), c_{v})$ and $T(\SF{}{u}(c_{y}), c_{w})$ but no fact of the form $R(\SF{}{u}(c_{y}), \dots)$
  in $\mathcal{O}(\R, \HeadChoice_1, \Trigger, h^\textit{uc}_\Trigger)$.
\end{example}

\citeauthor{RMFA} and \citeauthor{DMFA} introduce similar notions to $\star$-unblockability in \shortcite{RMFA} and \shortcite{DMFA}.
Here, we not only present a more general criterion (see Definition~\ref{definition:decidable-unblockability}), but a blueprint to produce more
comprehensive notions (see Definitions~\ref{def:semantic-over-approximation}, \ref{definition:over-approximation} and Lemmas~\ref{lemma:over-approximation-hc}, \ref{lemma:star-uc-over-approximations}).

\subsubsection{Propagating Unblockability}

A key feature of \textit{uc}/$\star$-unblockability is that these properties propagate across a \emph{reversible} constant-mappings: \cite{DMFA}

\begin{definition}\label{def:reversible-constant-mapping}
A constant mapping $g$ is a partial function mapping constants to terms.
For an expression \Expression, let $g(\Expression)$ be the expression that results from replacing all syntactic occurrences of every constant $c$ in the domain of $g$ with $g(c)$.

Consider a (possibly finite) set $\mathcal{T}$ of terms that contains every subterm of every $t \in \mathcal{T}$.
A constant mapping $g$ is \emph{reversible} for $\mathcal{T}$ if the following hold:
\begin{enumerate}
\item The function $g$ is defined for every constant in $\mathcal{T}$. \label{condition:reversibility1}
\item For every $t, s \in \mathcal{T}$ with $t \neq s$, we have that $g(t) \neq g(s)$. \label{condition:reversibility2}
\item For every constant $c \in \mathcal{T}$, every subterm $s$ of $g(c)$, and every functional term $u \in \mathcal{T}$; we have that $g(u) \neq s$. \label{condition:reversibility3}
\end{enumerate}
\end{definition}

\begin{lemma}\label{lem:unblk-propagates}
Consider a rule set \R, a head-choice \HeadChoice, an \R-trigger $\Trigger = \Tuple{\Rule, \Subs}$, and a constant mapping $g$ reversible for \Skeleton{\R}{\Trigger}.
If \Tuple{\Rule, g \circ \Subs} is an \R-trigger and $\Tuple{\Rule, \Subs}$ is \textit{uc}/$\star$-unblockable for \R [and $\HeadChoice$], then so is $\Tuple{\Rule, g \circ \Subs}$.
\end{lemma}
\begin{proof}
  We define $g^{-1}$ as follows:
  For a term $t$, let $g^{-1}(t) = s$ if there is a term $s$
  that occurs in $\Skeleton{\R}{\Trigger}$ with $g(s) = t$,
  $g^{-1}(t) = t$ otherwise if $t$ is a constant that does not occur in $\EI{\Constants}{\Skeleton{\R}{\Tuple{\Rule, g \circ \Subs}}}$ (i.e. $g^{-1}$ is the identity on fresh constants introduced by $h_{\dots}^{\textit{uc}}$),
  and $g^{-1}(t) = \star$ otherwise.
  Note that $g^{-1}$ is well-defined because $g$ is reversible (cond. \ref{condition:reversibility2})
  for $\Skeleton{\R}{\Trigger}$.

  Consider the sets $\F'$ and $\FAux'$ of all facts
  that can be defined using any predicate in \R and the constants in
  $\EI{\Constants}{\Skeleton{\R}{\Trigger}} \cup \{\star\}$
  and $\EI{\Constants}{\Skeleton{\R}{\Tuple{\Rule, g \circ \Subs}}} \cup \{\star\}$,
  respectively.
  Moreover, consider the fact sets:
  $\F = \BirthFacts{\R}{\Trigger} \cup \F'$ and
  $\FAux = \BirthFacts{\R}{\Tuple{\Rule, g \circ \Subs}} \cup \FAux'$.
  Also, let the functions $\Tuple{h_{\F}, h_{\FAux}}$ be from $\{\Tuple{h_\Trigger^\star, h_{\Tuple{\Rule, g \circ \Subs}}^\star}, \Tuple{h_\Trigger^\textit{uc}, h_{\Tuple{\Rule, g \circ \Subs}}^\textit{uc}}\}$.

  First claim: $g^{-1}(\FAux) \subseteq \F$.
  Since $g^{-1}(\FAux') \subseteq \F$ follows trivially,
  we only show
  $g^{-1}(\BirthFacts{\R}{t}) \subseteq \F$ for every
  $t \in \EI{\Terms}{\FAux}$ via induction over the structure of terms.
  If $t$ is a constant, then $g^{-1}(\BirthFacts{\R}{t}) = \emptyset$; hence, the base case trivially holds.
  Regarding the induction step, consider a term $t$ that is of the form $\SF{\RuleAux}{\ell, y}(\Vs)$:

  Let $\Vz$ be the list of existentially quantified variables
  in $\HF_\ell(\RuleAux)$.
  Let \SubsAux be a substitution with $\Frontier(\RuleAux)\SubsAux = \Vs$.
  Moreover, let $H = \HF_\ell(\RuleAux)\SubsAux$.
  We only need to show that $g^{-1}(H) \subseteq \F$
  to verify the induction step.
  We observe:
  If $g^{-1}(\SF{\RuleAux}{\ell, z}(\Vs))$ is functional
  for some $z \in \Vz$,
  then $g^{-1}(\SF{\RuleAux}{\ell, z'}(\Vs)) = \SF{\RuleAux}{\ell, z'}(g^{-1}(\Vs))$
  for each $z' \in \Vz$ ($\dagger$).
  Now, we perform a case by case analysis on $g^{-1}(t)$:
  If $g^{-1}(t)$ is a functional term,
  then $g^{-1}(H) = g^{-1}(\HF_\ell(\RuleAux)\SubsAux) = \HF_\ell(\RuleAux)(g^{-1} \circ \SubsAux) \subseteq \F$ (by $\dagger$).
  If $g^{-1}(t) \in \Constants \setminus \{\star\}$,
  then $g^{-1}(\SF{\RuleAux}{\ell, z}(\Vs))$
  must be a constant for every $z \in \Vz$ (by $\dagger$).
  Since $g$ is reversible for $\Skeleton{\R}{\Trigger}$ (cond. \ref{condition:reversibility3}),
  $g^{-1}(s)$ is also a constant (possibly $\star$)
  for every $s \in \Vs$. Hence, $g^{-1}(H) \subseteq \F' \subseteq \F$.
  If $g^{-1}(t) = \star$ and $g^{-1}(t')$ is a constant (or $\star$)
  for every $t' \in \EI{\Terms}{H}$, then $g^{-1}(H) \subseteq \F' \subseteq \F$.
  The remaining case of $g^{-1}(t) = \star$ and $g^{-1}(t') \notin \Constants$ for a $t' \in \EI{\Terms}{H}$,
  contradicts reversibility of $g$ (cond. \ref{condition:reversibility3})
  since then, there must be a constant $c$ with $t' \in \Subterms(g(c))$.

  Second claim:
  The set $g^{-1}(\mathcal{O}(\R, [\HeadChoice,] \Tuple{\Rule, g \circ \Subs}, h_{\FAux}))$ is a subset of $\mathcal{O}(\R, [\HeadChoice,] \Trigger, h_{\F})$.
  There exists a finite list of triggers $\Tuple{\RuleAux_{1}, \SubsAux_{1}}, \dots, \Tuple{\RuleAux_{m}, \SubsAux_{m}}$
  that yields $\mathcal{O}(\R, [\HeadChoice,] \Tuple{\Rule, g \circ \Subs}, h_{\FAux})$ from \FAux according to Definition~\ref{definition:over-approximation}.
  With the first claim as a base case, we can show via induction that the triggers
  $\Tuple{\RuleAux_{1}, g^{-1} \circ \SubsAux_{1}}, \dots, \Tuple{\RuleAux_{m}, g^{-1} \circ \SubsAux_{m}}$
  can be used in the construction of $\mathcal{O}(\R, [\HeadChoice,] \Trigger, h_{\F})$.

  We conclude that $\Tuple{\Rule, g \circ \Subs}$ is \textit{uc}/$\star$-unblockable for \R [and $\HeadChoice$] as follows:
  Suppose for a contradiction that $\Tuple{\Rule, g \circ \Subs}$ is not \textit{uc}/$\star$-unblockable.
  Then, \Rule is not datalog and $\Tuple{\Rule, g \circ \Subs}$ is obsolete for
  $\mathcal{O}(\R, \Tuple{\Rule, g \circ \Subs}, h_{\Tuple{\Rule, g \circ \Subs}}^{\star})$
  [resp. $\mathcal{O}(\R, \HeadChoice, \Tuple{\Rule, g \circ \Subs}, h_{\Tuple{\Rule, g \circ \Subs}}^{\textit{uc}})$].
  By the second claim above, we obtain
  that $\Trigger$ is obsolete for
  $\mathcal{O}(\R, \Trigger, h_{\Trigger}^{\star})$
  [resp. $\mathcal{O}(\R, \HeadChoice, \Trigger, h_{\Trigger}^{\textit{uc}})$].
  Hence, \Trigger is not \textit{uc}/$\star$-unblockable
  which contradicts the premise of the lemma.
\end{proof}


Condition~\ref{condition:reversibility2} in Definition~\ref{def:reversible-constant-mapping} admits an ``inverse'' of $g$ on term-level.
Lemma~\ref{lem:unblk-propagates} breaks without it:

\begin{example}
\label{exp:revInjectivityNecessary}
Consider the rule set $\R = \{(\ref{example:rev-2-1}\text{--}\ref{example:rev-2-4})\}$ and the head-choice $\HeadChoice_1$ mapping all rules to $1$.
\begin{align}
P(x, y) &\to \exists u. R(x, u) \land S(y, u) \label{example:rev-2-1} \\
R(x, y) &\to \exists v. T(y, v) \label{example:rev-2-2} \\
R(x, y) \land S(x, y) &\to T(y, x) \label{example:rev-2-3} \\
T(x, y) &\to P(y, y) \label{example:rev-2-4}
\end{align}
Consider the substitution $\Subs = [x / c_{x}, y / \SF{}{u}(c_{x}, c_{y})]$; the trigger $\Tuple{\eqref{example:rev-2-2}, \Subs}$, which is \textit{uc}-unblockable for \R and $\HeadChoice_1$; and the constant mapping $g$ that maps $c_x$ and $c_y$ to $\SF{}{v}(\SF{}{u}(c_{x}, c_{y}))$, which does not satisfy \eqref{condition:reversibility2} in Definition~\ref{def:reversible-constant-mapping} for $\mathcal{T} = \Skeleton{\R}{\Tuple{\eqref{example:rev-2-2}, \Subs}}$.
Indeed, $\Tuple{\eqref{example:rev-2-2}, g \circ \Subs}$ is not \textit{uc}-unblockable!
Intuitively, rule~\eqref{example:rev-2-3} ``blocks'' rule~\eqref{example:rev-2-2}
if rule~\eqref{example:rev-2-1} is applied with a substitution that maps $x$ and $y$ to the same term, which is the case for $g \circ \Subs$.
\end{example}


Lemma~\ref{lem:unblk-propagates} also breaks without condition \ref{condition:reversibility3}:

\begin{example}
Consider the head-choice $\HeadChoice_1$ mapping all rules to $1$, and the rule set \R containing the following:
\begin{align}
A(x) &\to \exists u. P(x, u) & &B(x) \to \exists v. Q(x, v) \nonumber \\
C(x) &\to \exists w. S(x, w) & &Q(x, y) \to T(x) \nonumber \\
P(x, y) &\to T(y) \lor \exists z. R(x, y, z) \label{rev3Crit}
\end{align}
The trigger $\Tuple{\eqref{rev3Crit}, [x / c, y / \SF{}{u}(d)]}$ is \textit{uc}-unblockable; the constant mapping $g$ that maps $c$ to $\SF{}{w}(\SF{}{v}(\SF{}{u}(d)))$ and $d$ to itself satisfies conditions \eqref{condition:reversibility1} and \eqref{condition:reversibility2} in Definition~\ref{def:reversible-constant-mapping} for $\mathcal{T} = \Skeleton{\R}{\Tuple{\eqref{rev3Crit}, [x / c, y / \SF{}{u}(d)]}}$.
Condition~\eqref{condition:reversibility3} is violated because $\SF{}{u}(d)$ is a subterm of $g(c)$ and $g(\SF{}{u}(d)) = \SF{}{u}(d)$.
Indeed, $\Tuple{\eqref{rev3Crit}, g \circ [x / c, y / \SF{}{u}(d)]}$ is not \textit{uc}-unblockable!
Intuitively, this is because the birth facts feature $Q(\SF{}{u}(d), \SF{}{v}(\SF{}{u}(d)))$,
which yields $T(\SF{}{u}(d))$, thus ``blocking'' the trigger $\Tuple{\eqref{rev3Crit}, [x / \SF{}{w}(\SF{}{v}(\SF{}{u}(d))), y / \SF{}{u}(d)]}$.
\end{example}

\section{Cyclicity Prefixes}
\label{section:cyclicity-prefixes}

Our high-level strategy to show that a rule set \R never-terminates is
to find a cyclicity sequence (see Definition~\ref{definition:cyclicity-sequence} and Theorem~\ref{theorem:cyclicity-sequence}),
which is challenging because it is infinite by definition.
Instead, we construct a \emph{cyclicity prefix} for \R, which is finite and still yields a cyclicity sequence.

Intuitively, a cyclicity prefix is a (finite) list of \textit{uc}-unblockable triggers that, if subsequently applied to a starting database \D, produce an isomorphic copy of \D that contains at least one new term. We can then repeat the prefix to obtain a cyclicity sequence.
To limit the number of starting databases,
we only consider minimal databases for that some trigger with a generating rule in \R is loaded.

\begin{definition}
The \emph{rule-database} of a rule \Rule is the database $\RuleDatabase{\Rule} = \BF(\Rule)\UCSubs$ where \UCSubs is a substitution that maps every variable $x$ to a fresh constant $c_x$ unique for $x$.
\end{definition}

Assume that rule \Rule can indeed be applied twice when starting on \RuleDatabase{\Rule}
and that the second application yields a cyclic term.
If the triggers applied in between the first and last application of \Rule are \textit{uc}-unblockable, then this finite sequence of triggers is a cyclicity prefix, which can be extended into a cyclicity sequence by applying Lemma~\ref{lem:unblk-propagates}.

\begin{definition}
\label{definition:cyclicity-prefix}
A \emph{cyclicity prefix} for a rule set \R, a head-choice \HeadChoice, and a rule \Rule is a (finite) list of \R-triggers $\TriggerSeq = \Tuple{\Rule_0, \Subs_0}$, $\ldots, \Tuple{\Rule_n, \Subs_n}$ such that:
\begin{itemize}
\item Both $\Rule_0 = \Rule$ and $\Subs_0 = \UCSubs$.
\item The sequence \TriggerSeq is loaded for \Tuple{\R, \RuleDatabase{\Rule}} and \HeadChoice.
\item Each trigger $\Tuple{\Rule_i, \Subs_i}$ with $1 \leq i \leq n$ is \textit{uc}-unblockable and $\Tuple{\Rule_{0}, \Subs_{0}}$ is g-unblockable.
\item Both $\Rule_n = \Rule$ and $\Output_{\HeadChoice}(\Tuple{\Rule_{n}, \Subs_{n}})$ features a \Rule-cyclic term; that is, a term $t$ that of the form $f(\Vs)$ with $f \in \EI{\Funs}{\Skolemise(\Rule)}$ and $f \in \EI{\Funs}{\Vs}$,
\item The constant mapping $g_\TriggerSeq$
      with $g_{\TriggerSeq} \circ \Subs_0 = \Subs_{n}$ is
      reversible for $\Skeleton{\R}{\Tuple{\Rule_{i}, g_{\TriggerSeq}^{j} \circ \Subs_{i}}}$ for every $1 \leq i \leq n$ and every $j \geq 0$.
      Note that $g_{\TriggerSeq}^{0}$ is the identity function over constants, and $g_{\TriggerSeq}^{i} = g_{\TriggerSeq} \circ g_{\TriggerSeq}^{i-1}$ for every $i \geq 1$.
\end{itemize}
\end{definition}

We can extend a cyclicity prefix such as \TriggerSeq above into an infinite sequence of triggers that are defined via composition with the constant-mapping $g_\TriggerSeq$;
afterwards, we show that this extension is a cyclicity sequence.
\begin{definition}
\label{definition:extend-prefix}
Given a \emph{cyclicity prefix} $\TriggerSeq = \Tuple{\Rule_0, \Subs_0}$, $\ldots,$ $\Tuple{\Rule_n, \Subs_n}$ for a rule set \R, a head-choice, and a rule in \R; let $\TriggerSeq^{\infty}$ be the (infinite) sequence $\Tuple{\Rule_0, \Subs_0}$, $\Tuple{\Rule_1, \Subs_1^1}$, $\ldots,$ $\Tuple{\Rule_n, \Subs_n^1}, \Tuple{\Rule_1, \Subs_1^2}, \ldots, \Tuple{\Rule_n, \Subs_n^2}, \ldots$ of \R-triggers with $\Subs_{i}^{j} = g_{\TriggerSeq}^{j-1} \circ \Subs_i$ for every $1 \leq i \leq n$ and $j \geq 1$.
\end{definition}

\begin{example}\label{exp:mainCycPrefix}
  The finite trigger sequence $\Tuple{\eqref{main1}, [x / c_{x}]}$, $\Tuple{\eqref{main2}, [x / \SF{}{v}(c_{x})]}$, $\Tuple{\eqref{main1}, [x / \SF{}{w}(\SF{}{v}(c_{x}))]}$ is a cyclicity prefix for the rule set $\{ \eqref{main1}, \eqref{main2} \}$, head-choice $\HeadChoice_1$, and \eqref{main1}.
\end{example}

\begin{theorem}
\label{theorem:cyclicity-prefix}
If $\TriggerSeq$ is a cyclicity prefix for a rule set \R, a head-choice \HeadChoice, and some $\Rule \in \R$; then $\TriggerSeq^\infty$ is a cyclicity sequence for $\Tuple{\R, \RuleDatabase{\Rule}}$ and \HeadChoice and hence, \R never-terminates.
\end{theorem}
\begin{proof}
  Assume that there is a cyclicity prefix $\TriggerSeq = \Tuple{\Rule_{0}, \Subs_{0}}, \dots,$ $\Tuple{\Rule_{n}, \Subs_{n}}$
  for \R, \Rule, and \HeadChoice; and consider the constant-mapping $g_{\TriggerSeq}$
  introduced in Definition~\ref{definition:cyclicity-prefix}.
  To prove Theorem~\ref{theorem:cyclicity-prefix},
  we show that $\TriggerSeq^{\infty}$ is a cyclicity sequence of the KB $\KB = \Tuple{\R, \RuleDatabase{\Rule}}$ and \HeadChoice.
  Namely, we argue that $\TriggerSeq^{\infty}$ is
  (a) loaded,
  (b) growing, and
  (c) g-unblockable.

  (a):
  We show that the trigger $\Tuple{\Rule_{i}, \Subs_{i}^{j}}$ is loaded in $\TriggerSeq^{\infty}$ for every $1 \leq i \leq n$ via induction on $j \geq 1$.
  The base case holds since $\TriggerSeq$ is loaded.
  The induction step from $j-1$ to $j$ holds since $g_{\TriggerSeq}(\F_{(j-1) \times n + i}(\KB, \HeadChoice, \TriggerSeq^{\infty}))$ is included in $\F_{j \times n + i}(\KB, \HeadChoice, \TriggerSeq^{\infty})$.


  (b): Since $\Output_{\HeadChoice}(\Tuple{\Rule_{n}, \Subs_{n}})$
  features a \Rule-cyclic term, there is some $x \in \Frontier(\Rule)$
  such that $\Subs_{n}(x) = g_{\TriggerSeq}(\Subs_{0}(x))$ is functional.
  Note that $\Subs_{0}(x)$ is a constant $c$,
  $\Depth(c) = 0$ and $\Depth(g_{\TriggerSeq}(c)) \geq 1$.
  Furthermore, $g_{\TriggerSeq}(c)$ features $c$ as a subterm.
  Hence, $g_{\TriggerSeq}^{k}(c) > g_{\TriggerSeq}^{k-1}(c)$ for every $k \geq 1$.
  But then, by construction of $\TriggerSeq^{\infty}$,
  $g_{\TriggerSeq}^{k}(c)$ occurs in $\F_{j}(\KB, \HeadChoice, \TriggerSeq^{\infty})$ for some $j \geq 0$.
  Thus, $\TriggerSeq$ must be growing.

  (c): We already know by assumption that $\Tuple{\Rule_{0}, \Subs_{0}}$
  is g-unblockable.
  We can show via induction over $j \geq 1$ that $\Tuple{\Rule_{i}, \Subs_{i}^{j}}$
  is \textit{uc}/$\star$-unblockable for every $1 \leq i \leq n$.
  For the base case with $j = 1$, the claim holds by assumption.
  For the induction step from $j$ to $j+1$,
  the claim follows from Lemma~\ref{lem:unblk-propagates}
  since $g_{\Trigger}$ is reversible for $\Skeleton{\R}{\Tuple{\Rule_{i}, \Subs_{i}^{j}}}$.
  By Lemma~\ref{lemma:deciding-unblockability},
  the sequence $\TriggerSeq^{\infty}$ is g-unblockable.
\end{proof}

\section{Novel Cyclicity Notions}
\label{section:cyclicity-notions}

Theorem~\ref{theorem:cyclicity-prefix} provides a blueprint to show non-termination of a rule set \R:
One simply has to show that \R admits a (finite) cyclicity prefix.
In this section, we present two different ways to do so, namely $\RPC$ and $\DRPC$, which we then use to define several cyclicity notions.

\subsubsection{Restricted Prefix Cyclicity}

We introduce \emph{restricted prefix-cyclicity} as the most general notion that we can define using our previous considerations:

\begin{definition}
\label{definition:rpc}
For a rule set \R, a head-choice \HeadChoice, and a rule $\Rule \in \R$;
let $\RPCF(\R, \HeadChoice, \Rule)$ be the fact set that includes the database $\RuleDatabase{\Rule}$,
the set $\Output_{\HeadChoice}(\Tuple{\Rule, \UCSubs})$,
and $\Output_\HeadChoice(\Trigger)$ for every \R-trigger $\Trigger = \Tuple{\RuleAux, \SubsAux}$ such that
\begin{itemize}
\item there are no cyclic terms in the range of \SubsAux,
\item the trigger \Trigger is loaded for $\RPCF(\R, \HeadChoice, \Rule)$,
\item the trigger \Trigger is \textit{uc}-unblockable for $\R$ and \HeadChoice, and
\item the substitution \Subs is injective if $\RuleAux = \Rule$.
\end{itemize}

A rule set \R is \emph{restricted prefix-cyclic} (\RPC)
if there is some head-choice \HeadChoice,
some (generating rule) $\Rule \in \R$,
and some \Rule-cyclic term that occurs in $\RPCF(\R, \HeadChoice, \Rule)$.
\end{definition}

The first restriction ensures that $\RPCF(\R, \HeadChoice, \Rule)$ is finite.
The second and third are necessary by Definition~\ref{definition:cyclicity-prefix}.
Note that $\Tuple{\Rule, \UCSubs}$ also needs to be g-unblockable but it is not \textit{uc}-unblockable by definition.
We show later that g-unblockability for $\Tuple{\Rule, \UCSubs}$ still follows if we can find a \Rule-cyclic term.
The fourth restriction ensures that we can find a reversible constant mapping for the cyclicity-prefix.
Example~\ref{exp:revInjectivityNecessary} shows a rule set that is terminating but
would be wrongly marked as \RPC if we omit the fourth restriction.

\begin{theorem}
\label{theorem:rpc-cyclicity}
If a rule set \R is \RPC, then it never-terminates.
\end{theorem}
\begin{proof}
  By Definition~\ref{definition:rpc},
  there exists a finite trigger sequence
  $\TriggerSeq = \Tuple{\Rule, \UCSubs}, \Trigger_{1} = \Tuple{\Rule_{1}, \Subs_{1}}, \dots, \Trigger_{n} = \Tuple{\Rule_{n}, \Subs_{n}}$
  that yields a (first) \Rule-cyclic term in
  $\RPCF(\R, \HeadChoice, \Rule)$
  with
  $\Rule_{n} = \Rule$ and
  no cyclic term in the image of every substitution $\Subs_{i}$ for $1 \leq i \leq n$,
  and
  a constant mapping $g_{\TriggerSeq}$ for that we find $g_{\TriggerSeq} \circ \UCSubs = \Subs_{n}$.
  Furthermore, \TriggerSeq is loaded for $\Tuple{\R, \RuleDatabase{\Rule}}$ and \HeadChoice,
  and the triggers $\Trigger_{1}, \dots, \Trigger_{n}$ are \textit{uc}-unblockable.

  To prove that \TriggerSeq is a cyclicity-prefix for \R, \HeadChoice, and \Rule,
  it only remains to show
  that (A) $g_{\TriggerSeq}$ is reversible for $\Skeleton{\R}{\Tuple{\Rule_{i}, g_{\TriggerSeq}^{j} \circ \Subs_{i}}}$ for each $1 \leq i \leq n$ and $j \geq 0$
  and that (B) $\Tuple{\Rule, \UCSubs}$ is g-unblockable.

  (A):
  Considering Definition~\ref{def:reversible-constant-mapping}, we show conditions \eqref{condition:reversibility1}, \eqref{condition:reversibility2}, and \eqref{condition:reversibility3}.
  Since $\EI{\Constants}{\Skeleton{\R}{\Tuple{\Rule_{i}, g_{\TriggerSeq}^{j} \circ \Subs_{i}}}}$
  may only feature constants from $\RuleDatabase{\Rule}$, \eqref{condition:reversibility1} holds.
  For \eqref{condition:reversibility2} and \eqref{condition:reversibility3}, we make the following observations:

  The substitutions of the triggers in $\TriggerSeq$ do not feature cyclic terms. Hence, for every constant $c$ in $\RuleDatabase{\Rule}$, the term
  $g_{\TriggerSeq}(c)$ does not feature
  nested function symbols from $\Skolemise(\Rule)$.
  \footnote{Consider the rule $\Rule = A(x) \to \exists y, z. R(x, y) \land S(x, z)$.
  Then, $\SF{}{y}(\SF{}{z}(c))$ features nested function symbols from $\Skolemise(\Rule)$ but $\SF{}{w}(\SF{}{y}(d), \SF{}{z}(c))$ does not (assuming $w$ occurs in another rule).}
  We show that, for any functional term $t$ in $\Skeleton{\R}{\Tuple{\Rule_{i}, g_{\TriggerSeq}^{j} \circ \Subs_{i}}}$ (for any $j$ and $i$),
  the term $g(t)$ features nested function symbols from $\Skolemise(\Rule)$:
  Every non-datalog trigger without functional terms in frontier positions
  is not \textit{uc}-unblockable.
  Hence, $t$ must have a subterm of the form $f(\Vc)$ such that $f$ occurs in $\Skolemise(\Rule)$
  and $\Vc = \UCSubs(\Frontier(\Rule))$.
  Also, for some $x \in \Frontier(\Rule_{n})$,
  $g_{\TriggerSeq}(\UCSubs(x))$
  is a functional term from
  $\Output_{\HeadChoice}(\Tuple{\Rule, \UCSubs})$.
  Thus,
  $f(g_{\TriggerSeq}(\Vc)) \in \Subterms(g(t))$ features nested function symbols from $\Skolemise(\Rule)$.


  By Definition~\ref{definition:rpc}, $g_{\TriggerSeq} \circ \UCSubs$ is injective
  and in turn $g_{\TriggerSeq}$ is injective on the constants in $\RuleDatabase{\Rule}$.
  We show \eqref{condition:reversibility2} that $g_{\TriggerSeq}(t) \neq g_{\TriggerSeq}(s)$
  for every $t, s$ in $\Skeleton{\R}{\Tuple{\Rule_{i}, g_{\TriggerSeq}^{j} \circ \Subs_{i}}}$
  with $t \neq s$:
  If $t$ and $s$ are constants, then $g_{\TriggerSeq}(t) \neq g_{\TriggerSeq}(s)$ since $g_{\TriggerSeq}$ is injective.
  If $t$ is a constant and $s$ is functional (or vice versa),
  then $g_{\TriggerSeq}(s)$ features nested function symbols from $\Skolemise(\Rule)$
  and $g_{\TriggerSeq}(t)$ does not (or vice versa) by the above observations.
  If $t$ and $s$ are functional terms of the form $f(\Vt)$ and $h(\Vs)$, respectively,
  with $f \neq h$; then $g_{\TriggerSeq}(t) \neq g_{\TriggerSeq}(s)$ holds.
  If $t$ and $s$ are functional terms (of finite depth)
  of the form $f(t_1, \ldots, t_n)$ and $f(s_1, \ldots, s_n)$, respectively;
  then $t_i \neq s_i$ for some $1 \leq i \leq n$ since $t \neq s$ and we can recurse into one of the cases for $t_{i}, s_{i}$.

  For \eqref{condition:reversibility3}, consider $c \in \EI{\Constants}{\Skeleton{\R}{\Tuple{\Rule_{i}, g_{\TriggerSeq}^{j} \circ \Subs_{i}}}} \subseteq \EI{\Constants}{\RuleDatabase{\Rule}}$
  and some $s \in \Subterms(g_{\TriggerSeq}(c))$.
  If there was a functional term $u \in \Skeleton{\R}{\Tuple{\Rule_{i}, g_{\TriggerSeq}^{j} \circ \Subs_{i}}}$ with $g_{\TriggerSeq}(u) = s$,
  we obtain a contradiction from the above observations,
  as $s \in \Subterms(g_{\TriggerSeq}(c))$
  does not feature nested function symbols
  from $\Skolemise(\Rule)$
  but
  $g_{\TriggerSeq}(u) = s$ does.
  Thus, (A) holds.

  (B): In the remainder of the proof, we show that $\Tuple{\Rule, \UCSubs}$ is g-unblockable for $\Tuple{\R, \RuleDatabase{\Rule}}$ and \HeadChoice.
  Suppose for a contradiction that $\Tuple{\Rule, \UCSubs}$ is not g-unblockable.
  We obtain the contradiction by showing that $\Tuple{\Rule, g_{\TriggerSeq} \circ \UCSubs}$
  is not \textit{uc}-unblockable.

  There exists a chase tree $\CT = \Tuple{V, E, \LF, \LT}$ for $\Tuple{\R, \RuleDatabase{\Rule}}$ and \HeadChoice such that
  $\Output_{\HeadChoice}(\Tuple{\Rule, \UCSubs}) \nsubseteq \LF(u)$ for each $u \in \Branch(\CT, \HeadChoice)$.
  Note that $\Tuple{\Rule, \UCSubs}$ is loaded for $\LF(v)$ for every $v \in \Branch(\CT, \HeadChoice)$ since it is loaded for $\RuleDatabase{\Rule}$.
  There must exists a (first) $w \in \Branch(\CT, \HeadChoice)$ such that $\Tuple{\Rule, \UCSubs}$ is obsolete for $\LF(w)$.
  Consider the path $v_{0}, \dots, v_{m}$ in \CT with $v_{0}$ the root and $v_{m} = w$.
  Let $\Tuple{\RuleAux_{1}, \SubsAux_{1}}, \dots, \Tuple{\RuleAux_{m}, \SubsAux_{m}} = \LT(v_{1}), \dots, \LT(v_{m})$.
  We aim to show that $h(\Output_{\HeadChoice}(\Tuple{\RuleAux_{i}, h \circ g_{\TriggerSeq} \circ \SubsAux_{i}})) \subseteq O$ for every $1 \leq i \leq m$
  where $h = h_{\Tuple{\Rule, g_{\TriggerSeq} \circ \UCSubs}}^{\textit{uc}}$ and $O = \mathcal{O}(\R, \HeadChoice, \Tuple{\Rule, g_{\TriggerSeq} \circ \UCSubs}, h)$.
  Then, we find that $\Tuple{\Rule, g_{\TriggerSeq} \circ \UCSubs}$ is not \textit{uc}-unblockable, i.e. the desired contradiction.

  First, we find that
  $h(g_{\TriggerSeq}(\LF(v_{0})) = h(\BF(\Rule)(g_{\TriggerSeq} \circ \UCSubs)) \subseteq O$
  by making use of the triggers in \TriggerSeq:
  We have that $h(\Output_{\HeadChoice}(\Tuple{\Rule, \UCSubs})) = \Output_{\HeadChoice}(\Tuple{\Rule, \UCSubs}) \subseteq \BirthFacts{\R}{\Tuple{\Rule, g_{\TriggerSeq} \circ \UCSubs}}$.
  Also, $h(\RuleDatabase{\Rule})$ is contained in the set of all facts that can be defined using
  any predicate and constants from $\EI{\Constants}{\Skeleton{\R}{\Tuple{\Rule, g_{\TriggerSeq} \circ \UCSubs}}} \cup \{\star\}$.
  Since \TriggerSeq is loaded, we can now show that $h(\Output_\HeadChoice(\Tuple{\Rule_{i}, h \circ \Subs_{i}})) \subseteq O$ for every $1 \leq i \leq n$.
  It is important to realize that
  $\Output_{\HeadChoice}(\Tuple{\Rule, g_{\TriggerSeq} \circ \UCSubs}) \neq \Output_{\HeadChoice}(\Tuple{\Rule_{i}, h \circ \Subs_{i}})$
  by the fact that $\Trigger_{n}$ is the first trigger to yield a \Rule-cyclic term.


  Now,
  $h(\Output_{\HeadChoice}(\Tuple{\RuleAux_{i}, h \circ g_{\TriggerSeq} \circ \SubsAux_{i}})) \subseteq O$ for $1 \leq i \leq m$
  can be verified given that $\Output_{\HeadChoice}(\Tuple{\RuleAux_{i}, h \circ g_{\TriggerSeq} \circ \SubsAux_{i}}) \neq \Output_{\HeadChoice}(\Tuple{\Rule, g_{\TriggerSeq} \circ \UCSubs})$. The latter is the case since $\Output_{\HeadChoice}(\Tuple{\Rule, \UCSubs}) \nsubseteq \LF(w)$.
  Therefore, $\Tuple{\Rule, g_{\TriggerSeq} \circ \UCSubs}$ is not \textit{uc}-unblockable; contradiction.
\end{proof}

In practice, it is infeasible to compute \RPC membership because of the exponential number of head-choices.
While this does not influence the complexity bounds (see Theorem~\ref{theorem:cyclicityComplexities}),
this exponential effort manifests often in practice.
Instead, we check $\RPC_{s}$ that considers far fewer head-choices but still explores a meaningful portion of the search space
by using each head-disjunct of each rule at least once:

\begin{definition}
\label{definition:rpcs}
For some $i \geq 1$, let $\HeadChoice_i$ be the head-choice such that, for every rule $\Rule$:
If $i \leq \BranchFactor(\Rule)$, then $\HeadChoice_i(\Rule) = i$.
Otherwise, $\HeadChoice_i(\Rule) = \BranchFactor(\Rule)$.
For a rule set $\R$, let $\BranchFactor(\R)$ be the smallest number such that $\BranchFactor(\R) \geq \BranchFactor(\Rule)$ for every $\Rule \in \R$.

A rule set \R is $\RPC_s$ if there is some $1 \leq i \leq \BranchFactor(\R)$, some (generating) $\Rule \in \R$, and some \Rule-cyclic term that occurs in $\RPCF(\R, \HeadChoice_i, \Rule)$.
\end{definition}

By Definitions~\ref{definition:rpc} and \ref{definition:rpcs}, a rule set \R is \RPC if it is $\RPC_s$.
Therefore, the following follows from Theorem~\ref{theorem:rpc-cyclicity}:

\begin{theorem}
\label{theorem:rpcs-cyclicity}
If a rule set is $\RPC_s$, then it never-terminates.
\end{theorem}





\subsubsection{Deterministic Restricted Prefix Cyclicity}

We introduce \emph{deterministic \RPC} as a less general version of \RPC; our goal here is to produce a notion that is similar to \RMFC \cite{RMFA} (see Section~\ref{section:related-work}).
It is therefore our baseline in the evaluation (see Section~\ref{section:evaluation}).
For example, the rule set $\{ \eqref{main1}, \eqref{main2} \}$ is \RPC but not \DRPC.

\begin{definition}
For a rule set \R and a deterministic rule $\Rule \in \R$, let $\DRPCF(\R, \Rule)$ be a fact set that includes the database $\RuleDatabase{\Rule}$, the set $\Output_1(\Tuple{\Rule, \UCSubs})$, and $\Output_1(\Trigger)$ for every deterministic \R-trigger $\Trigger = \Tuple{\RuleAux, \SubsAux}$ such that
\begin{itemize}
\item there are no cyclic terms in the range of \SubsAux,
\item the trigger \Trigger is loaded for $\DRPCF(\R, \Rule)$,
\item the trigger \Trigger is $\star$-unblockable for \R, and
\item the substitution \Subs is injective if $\RuleAux = \Rule$.
\end{itemize}

A rule set \R is \emph{deterministic restricted prefix-cyclic} (\DRPC) if there is some deterministic (generating) rule $\Rule \in \R$ and some \Rule-cyclic term that occurs in $\DRPCF(\R, \Rule)$.
\end{definition}

\begin{theorem}
\label{theorem:drpc-cyclicity}
If a rule set \R is $\DRPC$, it never-terminates.
\end{theorem}
\begin{proof}
  By Lemma~\ref{lemma:deciding-unblockability},
  every $\star$-unblockable trigger is also \textit{uc}-unblockable for every head-choice.
  Therefore, we find $\DRPCF(\R, \Rule) \subseteq \RPCF(\R, \HeadChoice, \Rule)$
  for every deterministic rule $\Rule \in \R$
  and every head-choice \HeadChoice.
  Hence, if \R is \DRPC, it is \RPC (even $\RPC_{s}$) and the claim follows by Theorem~\ref{theorem:rpc-cyclicity}.
\end{proof}

\subsubsection{Complexity}

The complexity of checking cyclicity is dominated by the double-exponential number of (non-cyclic) terms that may occur
during the check.
That is, checking \RPC, $\RPC_s$, or \DRPC requires at most a double-exponential number of steps of which each is possible in double-exponential time.
Hardness follows similarly to \MFA \cite[Theorem 8]{MFA}:
We extend $\Sigma_{3}$ to $\Sigma_{4}$
by adding a fresh atom $P_{\RuleAux}(\Vy)$ to the head of every $\RuleAux \in \Sigma_{3}$
where $\Vy$ is the list of all body variables in $\RuleAux$.
By this, we make sure that unblockability does not interfere with the original proof idea.
The set $\Omega$ is then defined as $\Sigma_{4} \cup \{ \Rule = R(w, x) \land B(x) \to \exists y. R(x, y) \land A(y) \}$.
We find that $\Tuple{\Sigma_{4}, \{ A(a) \}} \models B(a)$ iff $\Omega$ is \RPC, $\RPC_s$, or \DRPC.

\begin{theorem}\label{theorem:cyclicityComplexities}
Checking $\text{(D)RPC}_\text{(s)}$ is \DoubleExpTime-complete.
\end{theorem}

\section{Related Work}
\label{section:related-work}

Our main goal in this paper is to develop very general cyclicity notions for the disjunctive restricted chase.
To the best of our knowledge, the only such existing notion is \emph{restricted model faithful cyclicity} (\RMFC), which was introduced by \citeauthor{RMFA} in \shortcite{RMFA}.
While trying to extend \RMFC, 
we noticed that the proof of Theorem~11 in \shortcite{RMFA},
\footnote{Henceforth, we simply use \shortcite{RMFA} as an abbreviation of \cite{RMFA}.} 
which states that \RMFC rule sets do not terminate, is incorrect; correctness of the theorem remains open.

\begin{example}
\label{example:problem-rmfc-1}
Consider the rule set $\R = \{(\ref{rule:ce1}\text{--}\ref{rule:ce6})\}$.
\begin{align}
\FP{Cl}_1(x) &\wedge \FP{Cl}_2(y) \to \exists u . \FP{Red}(x, u) \wedge \FP{Red}(y, u) \label{rule:ce1} \\
\FP{Cl}_1(x) &\wedge \FP{Red}(x, z) \to \exists v . \FP{Gr}(x, v) \wedge \FP{Blu}(z, v) \label{rule:ce2} \\
\FP{Red}(y, z) &\wedge \FP{Blu}(z, w) \wedge \FP{Gr}(x, w) \to \FP{Gr}(y, y)  \label{rule:ce3} \\
\FP{Red}(y, z) &\wedge \FP{Blu}(z, w) \wedge \FP{Gr}(x, w) \to \FP{Blu}(z, y) \label{rule:ce4} \\
\FP{Red}(y, z) &\wedge \FP{Blu}(z, w) \wedge \FP{Gr}(x, w) \to \FP{Cl}_1(y) \label{rule:ce5} \\
\FP{Cl}_2(y) &\wedge \FP{Gr}(y, w) \to \FP{Cl}_2(w) \label{rule:ce6}
\end{align}

By Definition~11 in \shortcite{RMFA}, the rule set \R is \RMFC because the fact set $\F_\eqref{rule:ce1}$ features a \eqref{rule:ce1}-cyclic term.
As per the proof of Theorem~11 in \shortcite{RMFA}, the chase of \Tuple{\R, \Instance_\eqref{rule:ce1}} should ``contain infinitely many applications of \eqref{rule:ce1}''.
This is not the case; in fact, the result of the only chase tree of \Tuple{\R, \Instance_\eqref{rule:ce1}} is the set $\{ \F \}$ of fact sets where:
\begin{align*}
\F = &\{\FP{Cl}_1(c_x), \FP{Cl}_2(c_y), \FP{Red}(c_x, t), \FP{Red}(c_y, t), \FP{Gr}(c_x, s), \\
&\FP{Blu}(t, s), \FP{Gr}(c_x, c_x), \FP{Blu}(t, c_x), \FP{Gr}(c_y, c_y), \FP{Blu}(t, c_y)\}
\end{align*}
In the above, $t = f_{1, u}^\eqref{rule:ce1}(c_x, c_y)$ and $s = f_{1, v}^\eqref{rule:ce2}(c_x, t)$.
\end{example}

The problem stems from issues with Lemma~10 in \shortcite{RMFA}, which states that some triggers will eventually be applied if they are loaded for some vertex in the chase.
\begin{example}
\label{example:problem-rmfc-2}
By Definition~10 in \shortcite{RMFA}, a trigger such as $\Trigger = \Tuple{\eqref{rule:ce2}, [x / c_{y}, z / f_{1, u}^\eqref{rule:ce1}(c_x, c_y)]}$ with $c_x, c_y \in \Constants$ is unblockable for the rule set $\R = \{(\ref{rule:ce1}\text{--}\ref{rule:ce6})\}$.
One can verify that Lemma~10 in \shortcite{RMFA} does not hold for this trigger and the KB $\Tuple{\R, \{\FP{Cl}_1(c_x), \FP{Cl}_2(c_y)\}}$.
To do so, simply note that this trigger is loaded for the fact set \F defined at the end of Example~\ref{example:problem-rmfc-1}; however, $\F$ does not include $\Output_1(\Trigger)$.
Also, note that \Trigger is not \textit{uc}/$\star$-unblockable.
\end{example}

We have sought to ``repair'' \RMFC by introducing \DRPC.
We believe that both coincide for most real-world rule sets.

Another point of reference for us is our previous work \cite{DMFA}, where we have introduced \emph{Disjunctive Model Faithful Cyclicity (\DMFC)} for the (disjunctive) skolem chase.
We reuse many key ideas (also for the proofs) from this work, e.g. the main results for unblockability and reversiblity.
A necessary but straightforward change is the definition of obsoleteness.
While the idea of cyclicity sequences and prefixes was used in the proofs in spirit, a proper formalisation had not been presented.
Furthermore, \textit{uc}-unblockability was not considered for \DMFC.
Let us also stress that a cyclicity notion for the skolem chase is not a sufficient condition for restricted non-termination.
There are (many) rule sets that terminate for the restricted chase but not for the skolem chase.

\section{Evaluation}
\label{section:evaluation}

We have made available all evaluation materials online\footnote{\url{https://doi.org/10.5281/zenodo.8005904} \citeauthor{kr2023EvaluationMaterial}} including source code, rule sets, result files, and scripts used to obtain the counts.
In our experiments, we make use of a well-known sufficient condition for restricted chase termination to obtain an upper bound for the cyclicity notions.
\begin{definition}
A term is \emph{$k$-cyclic} for some $k \geq 1$ if it features $k + 1$ nested occurrences of the same function symbol.
For instance, $f(f(a))$ is 1-cyclic but $g(f(a), f(b))$ is not.

A rule set \R is $\RMFA_k$ for some $k \geq 1$ if there are no $k$-cyclic terms in the fact set $\RMFAF(\R)$, which is introduced in Definition~7 of \cite{RMFA}.
\end{definition}
\begin{theorem}
  $\RMFA_{k}$ rule sets (with $k \geq 1$) terminate.
\end{theorem}
The above result follows from the proof of Theorem~7 in \cite{RMFA}.

\subsubsection{Test Suite}
We consider rule sets from the evaluation of \cite{DMFA}, which were obtained from OWL ontologies via normalization and translation; see Section~6 in \cite{MFA} for more details.
OWL axioms with ``at-most restrictions'' and ``nominals'' are dropped because their translation requires the use of equality.
The ontologies come from the Oxford Ontology Repository (\textsf{OXFD})\footnote{\url{https://www.cs.ox.ac.uk/isg/ontologies/}} and the Manchester OWL Corpus (\textsf{MOWL}) \cite{mowl-corpus}.
We ignore rule sets without generating rules since these are trivially terminating.

\subsubsection{Results}
For every rule set \R in our test suite; we checked if $\R$ is $\RMFA_2$, \DRPC, and $\RPC_s$ using our implementations;
we ran each check with a 4h timeout on a cloud instance with 8 threads and 32GB of RAM (comparable to a modern laptop).
We present our results in Table~\ref{table:results}.
We split the rule sets that are purely deterministic ($\land$) from the ones containing disjunctions ($\lor$).
We further split by the number of generating rules $\#\exists$ and present the total number of rule sets $\#$ for each bucket.
For example, in the third row of the table, we indicate that there are 27 deterministic rule sets in \textsf{OXFD} with at least 20 and at most 99 generating rules of which 23 are $\RMFA_{2}$, 2 are \DRPC, and 3 are $\RPC_s$
When conidering $\RMFA_{2}$ together with $\DRPC$ or $\RPC_{s}$, the percentages of rule sets that cannot be characterised as either terminating or non-terminating
drop from \DRPC to $\RPC_{s}$. For \textsf{MOWL} $\land$, \textsf{OXFD} $\lor$, and \textsf{MOWL} $\lor$, the drops are
from $5\%$ to $3\%$, $37\%$ to $6\%$, and $45\%$ to $5\%$, respectively; for \textsf{OXFD} $\land$ the percentage is around $21\%$ for both.
Our improvements are significant on the datasets with disjunctions; for these, \RPC is considerably more general than \DRPC, (which we introduced as a replacement for \RMFC).

While many of the non-classified rule sets simply result from timeouts, there are $38$ rule sets in \textsf{OXFD} for which both $\RMFA_{2}$ and $\DRPC$ finished without capturing the rule set.
Analogously, with $\RPC_{s}$, we find $7$ rule sets.
For \textsf{MOWL} the numbers are $1505$ and $110$.
This indicates that there is still room for improvement on the theoretical side but also that timeouts are indeed a big issue, which happens often when we consider large datasets.

\begin{table}
\begin{tabular}{ c | c | c || c || c | c }
    & \#$\exists$ & \# tot. & $\RMFA_{2}$ & \DRPC & $\RPC_s$ \\
    \midrule

    \multirow{4}{*}{\rotatebox[origin=c]{90}{\textsf{OXFD} $\land$}}


    & 1--19 & 58 & 58 & 0 & 0+0 \\
    & 20--99 & 27 & 23 & 2 & 2+1 \\
    & 100--999 & 109 & 61 & 8 & 8+1 \\
    & \bf 1--999 & \bf 194 & \bf 142 & \bf 10 & \bf 10+2 \\

    \midrule

    \multirow{4}{*}{\rotatebox[origin=c]{90}{\textsf{MOWL} $\land$}}


    & 1--19 & 1139 & 866 & 239 & 239+12 \\
    & 20--99 & 269 & 228 & 27 & 27+5 \\
    & 100-999 & 363 & 271 & 46 & 46+21 \\
    & \bf 1--999 & \bf 1771 & \bf 1365 & \bf 312 & \bf 312+38\\

    \midrule
    \midrule

    \multirow{4}{*}{\rotatebox[origin=c]{90}{\textsf{OXFD} $\lor$}}


    & 1--19 & 37 & 32 & 0 & 0+5 \\
    & 20--99 & 18 & 4 & 7 & 7+7 \\
    & 100--999 & 147 & 8 & 13 & 13+20 \\
    & \bf 1--999 & \bf 102 & \bf 44 & \bf 20 & \bf 20+32 \\

    \midrule

    \multirow{4}{*}{\rotatebox[origin=c]{90}{\textsf{MOWL} $\lor$}}


    & 1--19 & 1361 & 806 & 48 & 48+405 \\
    & 20--99 & 894 & 196 & 171 & 171+496 \\
    & 100-999 & 1150 & 500 & 136 & 136+470 \\
    & \bf 1--999 & \bf 3405 & \bf 1502 & \bf 355 & \bf 355+1371
\end{tabular}
\caption{Restricted Chase Termination: Generating Rule Sets}\label{table:resultsRestricted}
\label{table:results}
\end{table}

\section{Conclusions and Future Work}
\label{section:conclusions}

We make three tangible contributions:
\FirstItem We define \RPC; a very general cyclicity notion tailored for rule sets with disjunctions.
\SecondItem We discovered problems with \RMFC and defined \DRPC as a ``repaired'' version of this notion.
\ThirdItem We present an evaluation to demonstrate the usefulness of our work.
Beyond these three, we believe that our efforts provide a framework for interesting future work.

\subsubsection{Extending Cyclicity Notions}
Despite the fact that \RPC is more general than existing criteria, there are many rule sets in our evaluation that remain open; that is, rule sets cannot be characterised as terminating or non-terminating.

Our work provides three different main strategies to achieve possible extensions.
The first one is to produce ``weaker'' over-approximations and thus a more general strategy to detect g-unblockability;
even g-unblockability itself can be relaxed.
The second is to generalize cyclicity prefixes; perhaps by checking loadedness from slightly different databases.
For instance, the rule set \R from Example~\ref{example:problem-rmfc-1} is neither \DRPC  nor \RPC (as intended)
but actually it is never-terminating; note that $\Tuple{\R, \{\FP{Cl}_{1}(c), \FP{Cl}_{2}(c)\}}$ does not admit finite chase trees.
Thus, even correctness of \RMFC remains an open problem.
The third one is to develop more comprehensive search strategies to find cyclicity prefixes; for instance, we can relax the condition in the first item of Definition~\ref{definition:rpc} to look a bit further.

\subsubsection{Explaining Cyclicity}

In many real-world use-cases, the existence of infinite universal models highlights a modelling mistake.
We can use \RPC and \DRPC as methods to explain the loss of termination.
For instance, a cyclicity prefix as defined in Section~\ref{section:cyclicity-prefixes} constitutes a small and clear explanation of one way of loosing termination.
In the future, we aim to automatically compute minimal sets of rules that can be removed (or added!) to deactivate a cyclicity prefix.

\section*{Acknowledgements}

Lukas is / has been funded
by Deutsche Forschungsgemeinschaft (DFG, German Research Foundation) in project 389792660 (TRR 248, Center for Perspicuous Systems),
by the Bundesministerium für Bildung und Forschung (BMBF, Federal Ministry of Education and Research) under European ITEA project 01IS21084 (InnoSale, Innovating Sales and Planning of Complex Industrial Products Exploiting Artificial Intelligence),
by BMBF and DAAD (German Academic Exchange Service) in project 57616814 (SECAI, School of Embedded and Composite AI),
and by the Center for Advancing Electronics Dresden (cfaed).

David is funded by the ANR project CQFD (ANR-18-CE23-0003).

\bibliographystyle{kr}
\bibliography{main}

\begin{tr}
    \newpage
    \appendix
    \section{Proof of Theorem~\ref{theorem:cyclicity-sequence}}

We ellaborate on the last part of the proof:

\begin{proof}[Proof (Last part extended).]
  To prove that $\F(\CT, \HeadChoice)$ is infinite,
  we verify that (A) $\bigcup_{i \geq 0} \F_i(\KB, \HeadChoice, \TriggerSeq)$ is infinite and
  that (B) $\F_i(\KB, \HeadChoice, \TriggerSeq) \subseteq \F(\CT, \HeadChoice)$ for every $i \geq 0$.
  \begin{itemize}
    \item Claim~(A) follows from the fact that \TriggerSeq is growing.
          This requirement implies that, for every $i \geq 0$, there is some $j > i$ such that $\F_i(\KB, \HeadChoice, \TriggerSeq) \subset \F_j(\KB, \HeadChoice, \TriggerSeq)$.
    \item We show Claim~(B) via induction on $i \geq 0$.
          The base case holds since $\F(\CT, \HeadChoice)$ includes $\D = \LF(v_{1})$.
          Regarding the induction step, consider some $i \geq 1$ and assume that the induction hypothesis holds; that is, that $\F(\CT, \HeadChoice)$ includes $\F_{i-1}(\KB, \HeadChoice, \TriggerSeq)$.
          Therefore, there is some $j \geq 1$ such that $\F_{i-1}(\KB, \HeadChoice, \TriggerSeq) \subseteq \LF(v_j)$ and hence, $\Trigger_i$ is loaded for $\LF(v_j)$.
          Since $\Trigger_i$ is g-unblockable, $\Output_{\HeadChoice}(\Trigger_i) \subseteq \LF(v_k)$ for some $k \geq 1$ and hence, $\F_i(\KB, \HeadChoice, \TriggerSeq) \subseteq \F(\CT, \HeadChoice)$ since $\LF(v_k) \subseteq \F(\CT, \HeadChoice)$.
  \end{itemize}
\end{proof}

\section{Proof of Theorem~\ref{theorem:g-unblockability-undecidable}}

\begin{proof}
We present a reduction from the undecidable problem of checking if a deterministic KB $\Tuple{\R, \D}$ entails a fact $\FP{P}(\vec{c})$ \cite{DBLP:conf/icalp/BeeriV81}.
Consider the head-choice $\HeadChoice_1$ that maps all rules to 1, and the rule set $\R' = \R \cup \{\Rule\}$ where $\Rule = \FP{P}(\vec{x}) \to \exists \vec{y} . \FP{P}(\vec{y})$; note that \R and $\R'$ are equivalent since \Rule is tautological.
We show that $\Tuple{\R, \D} \models \FP{P}(\vec{c})$ iff $\Trigger = \Tuple{\Rule, [\vec{x} / \vec{c}]}$ is not g-unblockable for $\R'$ and $\HeadChoice_1$:

The trigger $\Trigger$ is obsolete for a fact set when it is loaded.
Hence, if $\Tuple{\R, \D}$ entails $\FP{P}(\vec{c})$, then \Trigger is not g-unblockable for $\Tuple{\R', \D}$ and $\HeadChoice_1$.
If $\Tuple{\R, \D} \not\models \FP{P}(\vec{c})$,
then \Trigger is never loaded in any fact-label of a chase tree of $\Tuple{\R', \D}$ and therefore trivially g-unblockable.
\end{proof}

\section{Proof of Lemma~\ref{lemma:star-uc-over-approximations}}

\begin{proof}[Proof (Second condition of Definition~\ref{def:semantic-over-approximation})]
  Consider any $u \in \Branch(\CT, \HeadChoice)$
  in any
  chase tree $\CT = \Tuple{V, E, \LF, \LT}$ of any KB $\Tuple{\R, \D}$.
  Assuming $\Output_{\HeadChoice}(\Trigger) \nsubseteq \LF(u)$,
  we prove $h(\LF(u)) \subseteq \mathcal{O}(\R, \HeadChoice, \Trigger, h)$.

  Consider the path $u_{1}, \dots, u_{n}$ in \CT with the root $u_{1}$ and $u_{n} = u$.
  We show $h(\LF(u_i)) \subseteq \mathcal{O}(\R, \HeadChoice, \Trigger, h)$
  for every $1 \leq i \leq n$ via induction:
  The base case with $i = 1$ holds since the facts in $h(\D)$ are contained in the facts defined 
  by \eqref{condition:crit-inst} in Definition~\ref{definition:over-approximation}.
  For the induction step, consider $i \geq 2$:
  By induction hypothesis,
  $h(\LF(u_{i-1})) \subseteq \mathcal{O}(\R, \HeadChoice, \Trigger, h)$
  and hence, $\Tuple{\RuleAux, h \circ \SubsAux}$ is loaded
  for $\mathcal{O}(\R, \HeadChoice, \Trigger, h)$ where $\LT(u_i) = \Tuple{\RuleAux, \SubsAux}$.
  Since $\Output_\HeadChoice(\Trigger) \nsubseteq \LF(u)$:
  $\Output_{\HeadChoice}(\LT(u_{i})) \neq \Output_{\HeadChoice}(\Trigger)$.
  Hence, we also have $\Output_{\HeadChoice}(\Tuple{\RuleAux, h \circ \SubsAux}) \neq \Output_{\HeadChoice}(\Trigger)$.
  By Def.~\ref{definition:over-approximation}:
  $h(\Output_\HeadChoice(\Tuple{\RuleAux, h \circ \SubsAux})) \subseteq \mathcal{O}(\R, \HeadChoice, \Trigger, h)$.
  Therefore, $h(\LF(u_i)) \subseteq \mathcal{O}(\R, \HeadChoice, \Trigger, h)$ holds.
\end{proof}

\section{Proof of Lemma~\ref{lem:unblk-propagates}}

We elaborate on the first and second claims made in the main part of the proof.

\begin{proof}[Proof (First claim extended).]
We show that $g^{-1}(\FAux)$ is a subset of $\F$.
  Since $g^{-1}(\FAux') \subseteq \F$ follows trivially,
  we only have to show that
  $g^{-1}(\BirthFacts{\R}{t}) \subseteq \F$ for every
  $t \in \EI{\Terms}{\FAux}$; we do so via induction over the structure of terms.
  If $t$ is a constant, then $g^{-1}(\BirthFacts{\R}{t}) = \emptyset$; hence, the base case trivially holds.
  Regarding the induction step, consider a term $t$ that is of the form $\SF{\RuleAux}{\ell, y}(\Vs)$:

  \begin{enumerate}[a.]
    \item By ind.-hyp.: $g^{-1}(\BirthFacts{\R}{s}) \subseteq \F$ for every $s \in \Vs$.
    \item Let $\Vz$ be the list of existentially quantified variables
          in $\HF_\ell(\RuleAux)$.
          Let \SubsAux be a substitution with $\Frontier(\RuleAux)\SubsAux = \Vs$.
          Moreover, let $H = \HF_\ell(\RuleAux)\SubsAux$.
    \item By definition:
          $\BirthFacts{\R}{t} = H \cup \bigcup_{s \in \Vs} \BirthFacts{\R}{s}$.
    \item By (a) and (c): We only need to show that $g^{-1}(H) \subseteq \F$
          to verify the induction step.
          In fact, $g^{-1}(H) \subseteq \F$ follows from (f), (g), (h), and (i),
          which amount to a comprehensive case-by-case analysis.
    \item We observe:
          If $g^{-1}(\SF{\RuleAux}{\ell, z}(\Vs))$ is functional
          for some $z \in \Vz$,
          then $g^{-1}(\SF{\RuleAux}{\ell, z'}(\Vs)) = \SF{\RuleAux}{\ell, z'}(g^{-1}(\Vs))$
          for each $z' \in \Vz$.
    \item We show that $g^{-1}(H) \subseteq \F$ if $g^{-1}(t)$ is a functional term:
          In this case, $g^{-1}(H) = g^{-1}(\HF_\ell(\RuleAux)\SubsAux) = \HF_\ell(\RuleAux)(g^{-1} \circ \SubsAux) \subseteq \F$
          follows directly from (e).
    \item We show that $g^{-1}(H) \subseteq \F$ if $g^{-1}(t) \in \Constants \setminus \{\star\}$:
          If $g^{-1}(t) \in \Constants \setminus \{\star\}$,
          then $g^{-1}(\SF{\RuleAux}{\ell, z}(\Vs))$
          is a constant for every $z \in \Vz$ by (e).
          Since $g$ is reversible for $\Skeleton{\R}{\Trigger}$,
          $g^{-1}(s)$ is also a constant (possibly $\star$)
          for every $s \in \Vs$.
          Therefore, $g^{-1}(H) \subseteq \F' \subseteq \F$.
    \item If $g^{-1}(t) = \star$ and $g^{-1}(t')$ is a constant (or $\star$)
          for every $t' \in \EI{\Terms}{H}$, then $g^{-1}(H) \subseteq \F' \subseteq \F$.
    \item We show that assuming $g^{-1}(t) = \star$
          and $g^{-1}(t') \notin \Constants$
          for some $t' \in \EI{\Terms}{H}$ results in a contradiction:
          Note that $t'$ is necessarily functional because $g^{-1}(t')$ is.
          By (e), $t'$ can only occur in \Vs.
          Therefore, we have that, $t' \neq t$ is a subterm of $t$ such that $g^{-1}(t')$ is functional.

          At the same time, for $t$ to occur in $\Skeleton{\R}{\Tuple{\Rule, g \circ \Subs}}$,
          there needs to be a constant $c$ that occurs in the image of \Subs restricted to $\Frontier(\Rule)$
          such that $t$ occurs in $\BirthFacts{\R}{g(c)}$.

          Suppose for a contradiction that no such constant exists,
          i.e. there exists a functional term $u$ that occurs in the image of \Subs restricted to $\Frontier(\Rule)$
          such that $t$ occurs in $\BirthFacts{\R}{g(u)}$ but $t$
          does not occur in $\BirthFacts{\R}{g(u')}$
          for any subterm $u'$ of $u$ with $u' \neq u$.
          Since $g^{-1}(t)$ is not functional, $t$ must occur in $\BirthFacts{\R}{q}$ for a subterm $q$ of $g(u)$ with $q \neq g(u)$ by (e).
          But then, there exists a subterm $u'$ of $u$ with $u' \neq u$
          that occurs in the image of \Subs restricted to $\Frontier(\Rule)$
          with $g(u') = q$ since $u$ is functional.
          Since $t$ occurs in $\BirthFacts{\R}{g(u')}$,
          we obtain the desired contradiction and know that a constant $c$ of the desired form must exist.

          But then for $c \in \EI{\Constants}{\Skeleton{\R}{\Rule, \Subs}}$,
          we have $t' \in \Subterms(g(c))$ and
          there is a functional term $g^{-1}(t') \in \Skeleton{\R}{\Rule, \Subs}$
          with $g(g^{-1}(t')) = t'$, which contradicts reversibility of $g$.
  \end{enumerate}
\end{proof}

\begin{proof}[Proof (Second claim extended).]
Namely, we show that $g^{-1}(\mathcal{O}(\R, [\HeadChoice,] \Tuple{\Rule, g \circ \Subs}, h_{\FAux})) \subseteq \mathcal{O}(\R, [\HeadChoice,] \Trigger, h_{\F})$.
Consider a finite list of triggers $\Trigger_{1}, \dots, \Trigger_{m}$
  such that all of the following hold:

  \begin{itemize}
    \item $\mathcal{O}(\R, [\HeadChoice,] \Tuple{\Rule, g \circ \Subs}, h_{\FAux}) = \FAux \cup \bigcup_{i=1}^{m} h_{\FAux}(O_{i})$
          with $O_{i} = \bigcup \Output(\Trigger_{i})$ [resp. $O_{i} = \Output_{\HeadChoice}(\Trigger_{i})$].
    \item $\Trigger_{i}$ is loaded for $\FAux \cup \bigcup_{j=1}^{i-1} h_{\FAux}(O_{j})$.
    \item Let $\Tuple{\RuleAux_{i}, \SubsAux_{i}} = \Trigger_{i}$. We have $\RuleAux_{i} \neq \Rule$ or $\Output_{k}(\Tuple{\Rule, g \circ \Subs}) \neq \Output_{k}(\Trigger_{i})$
          for some $1 \leq k \leq \BranchFactor(\Rule)$. [Resp.: We have $\Output_{\HeadChoice}(\Tuple{\Rule, g \circ \Subs}) \neq \Output_{\HeadChoice}(\Trigger_{i})$.]
  \end{itemize}

  We show that $g^{-1}(\FAux \cup \bigcup_{j=1}^{i} h_{\FAux}(O_{j}))$ is a subset of $\mathcal{O}(\R, [\HeadChoice,] \Trigger, h_{\F})$
  via induction over $0 \leq i \leq m$.
  We have already shown the base case with $i = 0$, i.e. $g^{-1}(\FAux) \subseteq \F \subseteq \mathcal{O}(\R, [\HeadChoice,] \Trigger, h_{\F})$.

  Assume for the induction hypothesis that
  $g^{-1}(\FAux \cup \bigcup_{j=1}^{i} h_{\FAux}(O_{j}))$
  is a subset of $\mathcal{O}(\R, [\HeadChoice,] \Trigger, h_{\F})$ for some $i \geq 1$.
  To verify the induction step we only need to show that $g^{-1}(h_{\FAux}(O_{i+1})) \subseteq \mathcal{O}(\R, [\HeadChoice,] \Trigger, h_{\F})$.

  \begin{enumerate}[a.]
    \item For the trigger $\Tuple{\RuleAux_{i+1}, \SubsAux_{i+1}} = \Trigger_{i+1}$,
          we find that the fact sets
          $g^{-1}(h_{\FAux}(\bigcup \Output(\Tuple{\RuleAux_{i+1}, \SubsAux_{i+1}})))$ and $h_{\F}(\bigcup \Output(\Tuple{\RuleAux_{i+1}, g^{-1} \circ \SubsAux_{i+1}}))$ are equal.
          [Respectively: The fact sets $g^{-1}(h_{\FAux}(\Output_\HeadChoice(\Tuple{\RuleAux_{i+1}, \SubsAux_{i+1}})))$ and $h_{\F}(\Output_\HeadChoice(\Tuple{\RuleAux_{i+1}, g^{-1} \circ \SubsAux_{i+1}}))$ are equal.]
    \item By ind.-hypothesis, the trigger
          $\Tuple{\RuleAux_{i+1}, g^{-1} \circ \SubsAux_{i+1}}$
          is loaded for $\mathcal{O}(\R, [\HeadChoice,] \Trigger, h_{\F})$.
    \item Assume for a contradiction that
          $\RuleAux_{i+1} = \Rule$ and
          that for every $1 \leq k \leq \BranchFactor(\Rule)$, we obtain equality of $\Output_{k}(\Trigger)$ and $\Output_{k}(\Tuple{\RuleAux_{i+1}, g^{-1} \circ \SubsAux_{i+1}})$.
          [Resp.: Assume that $\Output_{\HeadChoice}(\Trigger) = \Output_{\HeadChoice}(\Tuple{\RuleAux_{i+1}, g^{-1} \circ \SubsAux_{i+1}})$.]
          Then, the output equalities also hold for $\Tuple{\Rule, g \circ \Subs}$ and $\Tuple{\RuleAux_{i+1}, g \circ g^{-1} \circ \SubsAux_{i+1}}$.
          Furthermore, the respective outputs of $\Tuple{\RuleAux_{i+1}, g \circ g^{-1} \circ \SubsAux_{i+1}}$ and $\Tuple{\RuleAux_{i+1}, \SubsAux_{i+1}}$ are equal.
          Therefore, we find a contradiction to the definition of $\Trigger_{1}, \dots, \Trigger_{m}$ above.
    \item By (a), (b), and (c): the induction step holds.
  \end{enumerate}
\end{proof}

\section{Proof of Theorem~\ref{theorem:cyclicity-prefix}}

We ellaborate on (a):

\begin{proof}[Proof ((a) extended).]
  We show via induction over $j \geq 1$ that $\Tuple{\Rule_{i}, \Subs_{i}^{j}}$ is loaded for every $1 \leq i \leq n$.
  The base case with $j=1$ holds since $\TriggerSeq$ is loaded.
  We show the induction step from $j$ to $j+1$.
  By induction hypothesis, $\Tuple{\Rule_{i}, \Subs_{i}^{j}}$ is loaded for $\F_{(j-1) \times n + i}(\KB, \HeadChoice, \TriggerSeq^{\infty})$.
  By construction, $\Tuple{\Rule_{i}, \Subs_{i}^{j+1}}$ is loaded for $g_{\TriggerSeq}(\F_{(j-1) \times n + i}(\KB, \HeadChoice, \TriggerSeq^{\infty}))$.
  Furthermore, we have $g_{\TriggerSeq}(\F_{(j-1) \times n + i}(\KB, \HeadChoice, \TriggerSeq^{\infty})) \subseteq \F_{j \times n + i}(\KB, \HeadChoice, \TriggerSeq^{\infty})$ by an inductive argument
  over the construction of $\TriggerSeq^{\infty}$ with the
  base case of having $g_{\TriggerSeq}(\F_{1}(\KB, \HeadChoice, \TriggerSeq^{\infty})) \subseteq \F_{n+1}(\KB, \HeadChoice, \TriggerSeq^{\infty})$.
  Note that the latter holds since $\Rule_{n} = \Rule$ and $g_{\TriggerSeq} \circ \Subs_{0} = \Subs_{n}$.
  Hence, $\Tuple{\Rule_{i}, \Subs_{i}^{j+1}}$ is loaded
  for $\F_{j \times n + i}(\KB, \HeadChoice, \TriggerSeq^{\infty})$,
  which yields the induction step.
\end{proof}

\section{Proof of Theorem~\ref{theorem:rpc-cyclicity}}

We ellaborate on the last part of (B):

\begin{proof}[Proof (last part of (B) extended).]
  We show in more detail that $h(\Output_{\HeadChoice}(\Tuple{\RuleAux_{i}, h \circ g_{\TriggerSeq} \circ \SubsAux_{i}})) \subseteq O$ for every $1 \leq i \leq m$.
First, we show
  $h(g_{\TriggerSeq}(\LF(v_{0})) = h(\BF(\Rule)(g_{\TriggerSeq} \circ \UCSubs)) \subseteq O$
  by making use of the triggers in \TriggerSeq:
  \begin{enumerate}[a.]
    \item We have that $h(\Output_{\HeadChoice}(\Tuple{\Rule, \UCSubs})) = \Output_{\HeadChoice}(\Tuple{\Rule, \UCSubs}) \subseteq \BirthFacts{\R}{\Tuple{\Rule, g_{\TriggerSeq} \circ \UCSubs}}$.
          Also, $h(\RuleDatabase{\Rule})$ is contained in the set of all facts that can be defined using
          any predicate and constants from $\EI{\Constants}{\Skeleton{\R}{\Tuple{\Rule, g_{\TriggerSeq} \circ \UCSubs}}} \cup \{\star\}$.
    \item By (a) and since \TriggerSeq is loaded,
          the trigger $\Tuple{\Rule_{i}, h \circ \Subs_{i}}$
          is loaded for $h(\F_{i-1}(\Tuple{\R, \RuleDatabase{\Rule} \cup \Output_{\HeadChoice}(\Tuple{\Rule, \UCSubs})}, \HeadChoice, \TriggerSeq))$
          for every $1 \leq i \leq n$.
    \item Since $\Tuple{\Rule, g_{\TriggerSeq} \circ \UCSubs}$ is the first trigger that yields a \Rule-cyclic term,
          $\Output_{\HeadChoice}(\Tuple{\Rule, g_{\TriggerSeq} \circ \UCSubs}) \neq \Output_{\HeadChoice}(\Tuple{\Rule_{i}, h \circ \Subs_{i}})$ for every $1 \leq i \leq n-1$.
    \item By (a), (b), (c), we find
          $h(\Output_\HeadChoice(\Tuple{\Rule_{i}, h \circ \Subs_{i}})) \subseteq O$
          for every $1 \leq i \leq n-1$, which with (b) concludes the claim.
  \end{enumerate}

  Now,
  $h(\Output_{\HeadChoice}(\Tuple{\RuleAux_{i}, h \circ g_{\TriggerSeq} \circ \SubsAux_{i}})) \subseteq O$ for $1 \leq i \leq m$.
  \begin{enumerate}[a.]
    \item We find that each trigger $\Tuple{\RuleAux_{i}, h \circ g_{\TriggerSeq} \circ \SubsAux_{i}}$ is loaded for
          $h(\BF(\Rule)(g_{\TriggerSeq} \circ \UCSubs)) \cup \bigcup_{j = 1}^{i-i} h(\Output_{\HeadChoice}(\Tuple{\RuleAux_{j}, h \circ g_{\TriggerSeq} \circ \SubsAux_{j}}))$.
    \item Since $\Output_{\HeadChoice}(\Tuple{\Rule, \UCSubs}) \nsubseteq \LF(w)$, we necessarily have that
          $\Output_{\HeadChoice}(\Tuple{\RuleAux_{i}, \SubsAux_{i}}) \neq \Output_{\HeadChoice}(\Tuple{\Rule, \UCSubs})$.
          But then also,
          $\Output_{\HeadChoice}(\Tuple{\RuleAux_{i}, h \circ g_{\TriggerSeq} \circ \SubsAux_{i}}) \neq \Output_{\HeadChoice}(\Tuple{\Rule, g_{\TriggerSeq} \circ \UCSubs})$.
    \item The claim follows from (a), (b), and the first claim shown in the previous enumeration.
  \end{enumerate}
\end{proof}

\section{Proof of Theorem~\ref{theorem:cyclicityComplexities}}

\begin{proof}
  \emph{Membership.}
  The number of rules to consider is linear in \R.
  The number of head-choices to consider is (at most) exponential in \R.
  The number of non-cyclic terms and therefore the number of triggers that need to be considered for any rule and head-choice is double-exponentially bounded in the size of \R.
  In particular, checking that a trigger is loaded and \textit{uc}/$\star$-unblockable takes at most double-exponential time.
  All together, checking $\text{(D)RPC}_\text{(s)}$ requires at most a double-exponential number of steps of which each is possible in double-exponential time.

  \emph{Hardness.}
  Following the hardness result for \MFA \cite[Theorem 8]{MFA},
  we use a reduction from the problem of conjunctive query entailment over weakly acyclic rule set \R (which is called $\Sigma$ in the original proof).
  Let $\R'$ be the weakly-acyclic rule set that results from \R such that
  $\R'' = \R' \cup \{ \Rule = R(w, x) \land B(x) \to \exists y. R(x, y) \land A(y) \}$
  is \MFA iff $\Tuple{\R', \{ A(a) \}} \not\models B(a)$ according to the construction by \citeauthor{MFA}
  In the original proof $\R'$ corresponds to $\Sigma_3$ \cite[Theorem 8]{MFA}.
  In turn, the rule set $\R''$ corresponds to $\Omega$ \cite[Lemma 7]{MFA}.
  Note that $\R'$ is weakly-acyclic and thus also \MFA and that no atom with $R$ occurs in $\R'$.
  We further extend every rule $\RuleAux \in \R'$ to obtain $\R'''$ by adding a fresh atom $P_{\RuleAux}(\Vy)$ to the head of $\RuleAux$
  where $\Vy$ is the list of all universally quantified variables in $\RuleAux$.
  Then, similar to $\R''$, we set $\R'''' = \R''' \cup \{ \Rule \}$.
  Again, $\R'''$ is weakly-acyclic and \MFA.
  Since $\R''''$ is deterministic, we consider the head-choice \HeadChoice that maps all rules to $1$.

  For $\text{(D)RPC}_\text{(s)}$, $\Output_{\HeadChoice}(\Tuple{\Rule, \UCSubs})$ already includes $R(c_x, \SF{\Rule}{1, y}(c_x))$ and $A(\SF{\Rule}{1, y}(c_x))$.
  Since $\R'''$ is \MFA, there are no cyclic terms in the $\text{(D)RPC}_\text{(s)}$ construction that do not feature a term from $\Output_{\HeadChoice}(\Tuple{\Rule, \UCSubs})$.
  Therefore, if $\Tuple{\R''', \{ A(a) \}} \not\models B(a)$, then no other trigger for \Rule is loaded
  and thus, $\R''''$ is not $\text{(D)RPC}_\text{(s)}$.
  For any other rule $\RuleAux \in \R'''$, the $\text{(D)RPC}_\text{(s)}$ construction fails because $\R'''$ is \MFA and
  \Rule can never be applied because the predicate $R$ only occurs in \Rule.
  Hence, $\R''''$ is not $\text{(D)RPC}_\text{(s)}$.
  Otherwise, if $\Tuple{\R''', \{ A(a) \}} \models B(a)$,
  we show that the sequence of $\R'''$-triggers that can derive $B(a)$ from $A(a)$ can be used in the construction
  of $\text{(D)RPC}_\text{(s)}$.
  Since $R$ does not occur in $\R'''$ and $B$ only occurs in a rule head in $\R'''$,
  $A(\SF{\Rule}{1, y}(c_x))$ is the only usable fact when starting to construct the chase derivation for $\text{(D)RPC}_\text{(s)}$.
  Hence, every trigger that becomes loaded in the contruction features a functional term in the image of its substitution.
  Since every rule in $R'''$ contains a head-atom featuring all universal variables, each of the loaded triggers is also \textit{uc}/$\star$-unblockable.
  Since $\R'''$ is \MFA, every loaded trigger does not feature cyclic terms.
  By that, we obtain, if $\Tuple{\R''', \{ A(a) \}} \models B(a)$, then $B(\SF{\Rule}{1, y}(c_x))$ occurs in the construction of $\text{(D)RPC}_\text{(s)}$.
  The trigger $\Tuple{\Rule, [w / c_x, x / \SF{\Rule}{1, y}(c_x)]}$ has an injective substitution.
  Therefore we obtain a \Rule-cyclic term in $A(\SF{\Rule}{1, y}(\SF{\Rule}{1, y}(c_x)))$
  and thus, $\R''''$ is $\text{(D)RPC}_\text{(s)}$.
\end{proof}

\end{tr}

\end{document}